\documentclass{fundam}

\usepackage{amsmath}
\usepackage{amssymb}

\usepackage{todonotes}

\usepackage{tikz}

\usepackage{xspace}

\usepackage[shortcuts]{extdash}

\usepackage{listings}

\usetikzlibrary{positioning}
\usetikzlibrary{decorations.pathreplacing}
\usetikzlibrary{automata,positioning}
\usetikzlibrary{decorations.pathmorphing}
\usetikzlibrary{decorations.markings}
\usetikzlibrary{decorations}
\usetikzlibrary{arrows}
\usetikzlibrary{patterns}
\usetikzlibrary{calc}
\usetikzlibrary{shapes}
\usetikzlibrary{fadings,	shadings}

\tikzstyle{ubrace} = [draw, thick, decoration={brace, mirror, raise=0.0cm}, decorate,
    every node/.style={anchor=north, yshift=-0.1cm}]
\tikzstyle{rbrace} = [draw, thick, decoration={brace, mirror, raise=0.0cm}, decorate,
    every node/.style={anchor=west, xshift= 0.1cm}]

\tikzstyle{obrace} = [draw, thick, decoration={brace, raise=0.0cm}, decorate,
    every node/.style={anchor=south, yshift= 0.1cm}]
\tikzstyle{lbrace} = [draw, thick, decoration={brace, raise=0.0cm}, decorate,
    every node/.style={anchor=east, xshift=-0.1cm}]

\newcommand{\boldclass}[3]{\ensuremath{\mathbf{#1}^{#2}_{#3}}}

\newcommand{\bsigma}[1]{\boldclass{\Sigma}{0}{#1}}
\newcommand{\bpi}[1]{\boldclass{\Pi}{0}{#1}}
\newcommand{\bdelta}[1]{\boldclass{\Delta}{0}{#1}}

\newcommand{\asigma}[1]{\boldclass{\Sigma}{1}{#1}}

\newcommand{\trans}[1]{\xrightarrow{#1}}

\newcommand{\eqdef}{=_{\mathrm{df}}}

\newcommand{\lang}{\mathrm{L}\xspace}

\newcommand{\init}{\mathrm{I}\xspace}

\newcommand{\run}{\mathrm{run}\xspace}
\newcommand{\acc}{\mathrm{acc}\xspace}

\newcommand{\IFinx}{\mathrm{IF}_{\mathrm{inf}\xspace}}

\newcommand{\lex}{<_{\mathrm{lex}}}

\newcommand{\lexeq}{\leq_{\mathrm{lex}}}

\newcommand{\inx}{<_{\mathrm{inf}}}
\newcommand{\inxg}{>_{\mathrm{inf}}}
\newcommand{\inxeq}{\leq_{\mathrm{inf}}}

\newcommand{\verA}{v}
\newcommand{\verB}{x}

\newcommand{\orrA}{o}

\newcommand{\finA}{u}

\newcommand{\infA}{\alpha}

\newcommand{\runA}{r}

\newcommand{\Aa}{\mathcal{A}\xspace}
\newcommand{\Cc}{\mathcal{C}\xspace}
\newcommand{\Tt}{\mathcal{T}\xspace}
\newcommand{\Mm}{\mathcal{M}\xspace}

\newcommand{\dL}{\mathtt{{\scriptstyle L}}}
\newcommand{\dM}{\mathtt{{\scriptstyle M}}}
\newcommand{\dR}{\mathtt{{\scriptstyle R}}}

\newcommand{\w}{\omega}

\newcommand{\N}{\mathbb{N}\xspace}

\newcommand{\restr}{{\upharpoonright}}

\newcommand{\borapxi}{{\bf\Sigma}^{0}_{\xi}}

\newcommand{\bormpxi}{{\bf\Pi}^{0}_{\xi}}

\newcommand{\borel}{{\bf\Delta}^{1}_{1}}

\newcommand{\om}{\omega}
\newcommand{\Si}{\Sigma}
\newcommand{\Sis}{\Sigma^\star}
\newcommand{\Sio}{\Sigma^\omega}
\newcommand{\nl}{\newline}
\newcommand{\lra}{\leftrightarrow}
\newcommand{\fa}{\forall}
\newcommand{\ra}{\rightarrow}

\newcommand{\Ga}{\Gamma}
\newcommand{\Gas}{\Gamma^\star}
\newcommand{\Gao}{\Gamma^\omega}
\usepackage{url}

  \usepackage{hyperref}

\begin{document}

\setcounter{page}{243}
\publyear{2021}
\papernumber{2088}
\volume{183}
\issue{3-4}

 \finalVersionForARXIV

\title{On the Expressive Power of  Non-deterministic  \\  and  Unambiguous Petri Nets over Infinite Words}

\author{Olivier Finkel\thanks{Address for correspondence: CNRS et Université de Paris, UFR de Math{\'e}matiques case 7012,
75205 Paris Cedex 13, France}
 \\
Institut de Math\'ematiques de Jussieu - Paris Rive Gauche \\
CNRS, Universit\'e de Paris,  Sorbonne Universit\'e, Paris, France.\\
Olivier.Finkel@math.univ-paris-diderot.fr
\and Micha{\l} Skrzypczak\thanks{Author supported by Polish National Science Centre
grant 2016/22/E/ST6/00041.}\\
Institute of Informatics \\
University of Warsaw \\ Banacha 2, 02-097 Warsaw, Poland\\
mskrzypczak@mimuw.edu.pl} \maketitle

\runninghead{O. Finkel and M. Skrzypczak}{On the Expressive Power of Petri Nets over Infinite Words}

\begin{abstract}
We prove that $\om$-languages of (non-deterministic) Petri nets and $\om$-languages of (non-deterministic) Turing machines have the same topological complexity:
the Borel and Wadge hierarchies of the class of  $\om$-languages of (non-deterministic) Petri nets are equal to the Borel and Wadge hierarchies of the class of $\om$-languages of (non-deterministic) Turing machines.  We also show  that it is highly undecidable to determine the topological complexity of a Petri net $\om$\=/language.  Moreover, we infer from the proofs of the above results that the equivalence and the inclusion problems for $\om$-languages of Petri nets are $\Pi_2^1$-complete, hence also highly undecidable.

Additionally, we show that the situation is quite the opposite when considering unambiguous Petri nets, which have the semantic property that at most one accepting run exists on every input. We provide a procedure of determinising them into deterministic Muller counter machines with counter copying. As a consequence, we entail that the $\omega$-languages recognisable by unambiguous Petri nets are $\bdelta{3}$ sets.
\end{abstract}


\section{Introduction}

In the sixties, B\"uchi was the first to study acceptance of infinite words by finite automata with the now called  B\"uchi  acceptance condition, in order to prove the decidability of the monadic second order theory of one successor over the integers.  Since then there has been a lot of work on regular $\omega$\=/languages, accepted by  B\"uchi automata, or by some other variants of automata over infinite words, like Muller or Rabin automata, see \cite{Thomas90,Staiger97,PerrinPin}. The acceptance of infinite words by other finite machines, like pushdown automata, counter automata, Petri nets, Turing machines, \ldots, with various acceptance conditions, has also been studied, see  \cite{Staiger97,eh,CG78b,valk1983infinite,Staiger93}.

The Cantor topology is a very natural topology on the set $\Si^\omega$ of infinite words over a finite alphabet~$\Si$ which is induced by the prefix metric. Then a way to study the complexity of languages of infinite words accepted by finite machines is to study    their topological complexity and firstly to locate them with regard to the Borel and the projective hierarchies \cite{Thomas90,eh,LescowThomas,Staiger97}.

Every $\omega$-language accepted by a deterministic B\"uchi automaton is a $\bpi{2}$-set.
On the other hand, it follows from McNaughton's Theorem that every regular $\omega$-language is accepted by a deterministic Muller automaton, and thus is a Boolean combination of $\omega$-languages accepted by deterministic B\"uchi automata. Therefore
every regular $\omega$-language is a $\bdelta{3}$-set. Moreover, Landweber proved that
the Borel complexity of any $\omega$-language accepted by a Muller or B\"uchi automaton can be effectively  computed (see \cite{Landweber69,PerrinPin}).
In a similar way, every $\omega$-language accepted by a deterministic Muller Turing machine, and thus also by any Muller deterministic finite machine is a $\bdelta{3}$-set, \cite{eh,Staiger97}.

The Wadge hierarchy is a great refinement of the Borel hierarchy, firstly defined by Wadge via reductions by continuous functions \cite{Wadge83}. The trace of the Wadge hierarchy on the regular $\omega$-languages is called the Wagner hierarchy. It has been completely described by Klaus Wagner in \cite{Wagner79}. Its length is the ordinal $\omega^\omega$. Wagner gave an automaton-like characterisation of this hierarchy, based on the notions of chain and superchain, together with an algorithm to compute the Wadge (Wagner) degree of any given regular $\omega$-language, see also \cite{Selivanov98,Selivanov08m,Selivanov08t,Simonnet92}.

The Wadge hierarchy of deterministic context-free $\omega$-languages was determined by Duparc in~\cite{Duparc03,dfr}. Its length is the ordinal $\omega^{(\omega^2)}$.
We do not know yet whether this hierarchy is decidable or not. But the Wadge hierarchy induced by deterministic partially blind 1-counter automata was described in an effective way in \cite{Fin01csl}, and other partial decidability results were obtained in \cite{Fin01a}. Then, it was proved in \cite{finkel2006borel} that the Wadge  hierarchy of  $1$-counter or context-free $\omega$-languages and the Wadge hierarchy of effective analytic sets (which form the class of all the $\omega$-languages accepted by non-deterministic Turing machines) are equal. Moreover, similar results hold about the Wadge hierarchy of infinitary rational relations accepted by $2$-tape B\"uchi automata, \cite{Fink-Wd}. Finally, the Wadge hierarchy of $\omega$-languages of deterministic
Turing machines was determined by Selivanov in \cite{Selivanov03b}.

We consider in this paper acceptance of infinite words by Petri nets. Petri nets are used for the description of distributed systems \cite{esparza1998decidability,rozenberg2004lectures,H-petrinet-diaz,hopcroft_vas,LerouxS19}, and form a very important  mathematical model in Concurrency Theory that has been developed for general concurrent computation. In the context of Automata Theory, Petri nets may be defined as
(partially) blind multicounter automata, as explained in \cite{valk1983infinite,eh,Gre78}.
First, one can distinguish between the places of a given Petri net by dividing them  into the bounded ones (the number of tokens in such a place at any time is uniformly bounded) and the unbounded ones. Then each unbounded place may be seen as a partially blind counter, and the tokens in the bounded places determine the state of the partially blind multicounter automaton that is equivalent to the initial Petri net. The transitions of the Petri net may then be seen as the finite control of the partially blind multicounter automaton and the labels of these transitions are then the input symbols.

The infinite behaviour of Petri nets (i.e.,~Petri nets running over $\omega$\=/words) was first studied by Valk \cite{valk1983infinite} and by Carstensen in the case of deterministic Petri nets \cite{carstensen1988infinite}. The topological complexity of $\om$\=/languages of {\it deterministic} Petri nets was completely determined. They are ${\bf \Delta}^0_3$-sets and their Wadge hierarchy has been determined by Duparc, Finkel, and Ressayre in~\cite{DFR4}; its length is the ordinal $\om^{\om^2}$. On the other side, Finkel and Skrzypczak proved in~\cite{FS14} that there exist $\bsigma{3}$-complete, hence non $\bdelta{3}$, $\om$-languages accepted by {\it non-deterministic} one-partially-blind-counter B\"uchi automata.

A next goal was to  understand the expressive power of {\it non-determinism}  in Petri nets, i.e.~blind multicounter automata. The assumption of blindness is important here, as it is already known that $\om$\=/languages accepted by (non-blind) one-counter B\"uchi automata have the same topological complexity as $\om$-languages of Turing machines~\cite{finkel2006borel}. However, the non-blindness of the counter, i.e.~the ability to use the zero-test of the counter, was essential in the proof of this result.

The first author proved in \cite{Fin17PN,Finkel20} that $\om$-languages of {\it non-deterministic} Petri nets and effective analytic sets have the same topological complexity. More precisely the Borel and Wadge hierarchies of the class of  $\om$-languages of  Petri nets are equal to the Borel and Wadge hierarchies of the class of effective analytic sets. The proof is based on a simulation of a  given real time $1$-counter (with zero-test) B\"uchi automaton $\mathcal{A}$ accepting $\om$-words $x$ over an alphabet $\Si$ by a real time 4-blind-counter B\"uchi automaton $\mathcal{B}$ reading some special codes $h(x)$ of  the words $x$. In particular,  for each non-null recursive ordinal $\alpha < \om_1^{{\rm CK}} $ there exist some $\bsigma{\alpha}$-complete and some $\bpi{\alpha}$-complete $\om$-languages of Petri nets, and the supremum of the set of Borel ranks of $\om$\=/languages of Petri nets is the ordinal $\gamma_2^1$, which is strictly greater than the first non-recursive ordinal~$\om_1^{{\rm CK}}$. Moreover it is proved in \cite{Fin17PN,Finkel20} that  it is highly undecidable to determine the topological complexity of a Petri net $\om$-language. Moreover,  it is inferred from the proofs of the above results that also the equivalence and the inclusion problems for $\om$-languages of Petri nets  are $\Pi_2^1$-complete, hence also highly undecidable.

A particular instance of the above construction of simulation is for the Wadge degrees of analytic sets $\asigma{1}$ --- it follows that there exists a real time 4-blind-counter B\"uchi automaton recognising a $\asigma{1}$-complete $\om$-language, but also that it is consistent with the axiomatic system  {\bf ZFC}  of set theory that there exists some non-Borel non-$\asigma{1}$-complete $\om$-language of 4-blind-counter B\"uchi automata, \cite{Fin17PN,Finkel20}.

The first author also  proved in \cite{Fin17PN}  that the determinacy of Wadge games between two players in charge of $\om$-languages of Petri nets  is equivalent to the effective analytic (Wadge) determinacy, and thus  is not provable in  the axiomatic system {\bf ZFC}.

The second author independently proved in \cite{S18} that  only one blind counter is enough to obtain a non-Borel $\om$-language --- he exhibits  an example of a real time 1-blind-counter B\"uchi automaton that also recognises a $\asigma{1}$-complete $\om$-language.

In the present journal paper we gather results of both authors from \cite{Fin17PN,Finkel20} and \cite{S18}  with some complements which did not appear in these papers:

 In particular, we  prove that there exists a Petri net accepting an $\om$-language $L$ such that $L$ is a Borel  ${\bf \Pi}^0_2$-set in one model of {\bf ZFC} and non-Borel in another model of {\bf ZFC} (Theorem \ref{top-zfc2}). We also show that it is undecidable to determine the topological complexity of the $\om$-language accepted by  a given $1$-blind counter automaton and, as a consequence, we prove several other undecidability results for   $\om$-languages accepted by   $1$-blind counter automata (Theorems \ref{undec-borel}, \ref{undec-effective-borel},  \ref{undec-other-results}, and \ref{undec-unambiguity}).

We also study the important case of unambiguous Petri nets. In that case, we provide a determinisation procedure for unambiguous blind counter automata that constructs an equivalent deterministic Muller counter machine with zero tests and counter copying, hence also an equivalent deterministic Muller Turing machine.
Its determinism already guarantees tight bounds on the topological complexity of unambiguous blind counter automata: their $\omega$-languages belong to the Borel class $\bdelta{3}$.
Notice that the blindness of the counters is crucial here since one can obtain some  ${\bf \Sigma}^0_3$-complete and  ${\bf \Pi}^0_3$-complete $\om$-languages accepted by unambiguous  $1$-counter automata with zero-test \cite{Fin03b}. The topological complexity of  $\omega$-languages  of unambiguous Petri nets may also been compared with the complexity of  $\omega$-languages  of unambiguous Turing machines: they form the class of effective $\Delta^1_1$-sets which contains ${\bf \Sigma}^0_\alpha$-complete sets and ${\bf \Pi}^0_\alpha$-complete sets
for each recursive ordinal $\alpha < \om_1^{{\rm CK}} $  \cite{Fin-ambTM}.

This paper is an extended journal version of both the paper \cite{Finkel20} of the first author  which  appeared in the Proceedings of the 41st International Conference on Application and Theory of Petri Nets and Concurrency, Petri Nets 2020,  which took place virtually in Paris on June 2020, and of the paper \cite{S18} of the second author which appeared in the Proceedings of the 12th International Conference on Reachability Problems, which took place in  Marseille on September  2018.

\medskip
The paper is organised as follows. In Section~\ref{sec:basic} we review the notions of (blind) counter automata and $\om$-languages. In Section~\ref{sec:topology} we recall notions of topology, and  the Borel and Wadge hierarchies on a Cantor space. Section~\ref{sec:wadge} is devoted to the main result of that work: the simulation construction that provides $4$-blind counter B\"uchi automata for levels of the Wadge hierarchy occupied by $1$-counter B\"uchi automata. Based on this construction,  we show that  the topological or arithmetical complexity of a  Petri net
$\om$-language is highly undecidable in Section~\ref{sec:high-undec-top}. The equivalence and the inclusion problems for $\om$-languages of Petri nets  are shown to be $\Pi_2^1$-complete in Section~\ref{sec:high-undec-equiv}. Section~\ref{sec:wadge-determ} is devoted to a study of determinacy of Wadge games with winning conditions given by $\omega$\=/languages recognised by blind counter automata. The example of a $\asigma{1}$-complete $\omega$-language recognised by a $1$-blind counter B\"uchi automaton is given in Section~\ref{sec:non-borel}. Section~\ref{sec:inherent-nondet}  is devoted to consequences of the high topological complexity of $\omega$\=/languages recognisable by Petri nets with regard to inherent non-determinism and ambiguity.  Section~\ref{sec:determinisation} provides a determinisation construction based on the assumption that a given automaton is unambiguous. Concluding remarks are given in Section~\ref{sec:conclusions}. An additional Appendix~\ref{app:construction} contains an extensive explanation of the technical construction from Section~\ref{sec:determinisation}.

\section{Basic notions}
\label{sec:basic}

We assume the reader to be familiar with the theory of formal ($\om$-)languages \cite{Staiger97,PerrinPin}. We recall the usual notations of formal language theory.

If $\Si$ is a finite alphabet, a {\it non-empty finite word} over $\Si$ is any sequence $\finA=a_1\ldots a_k$, where $a_i\in\Sigma$ for $i=1,\ldots ,k$, and $k$ is an integer $\geq 1$. The {\it length} of $\finA$ is $k$, denoted by $|\finA|$. The {\it empty word}  is denoted by $\epsilon$; its length is $0$. $\Sis$ is the {\it set of finite words} (including the empty word) over $\Sigma$, and we denote $\Si^+\eqdef \Sis \setminus \{\epsilon\}$.
A (finitary) {\it language} $V$ over an alphabet $\Sigma$ is a subset of $\Sis$.

The {\it first infinite ordinal} is $\om$. An $\om$-{\it word} over $\Si$ is an $\om$-sequence $a_1 \ldots a_n \ldots$, where for all integers $ i\geq 1$, ~
$a_i \in\Sigma$.  When $\sigma=a_1 \ldots a_n \ldots$ is a finite word of length at least $n$ or an $\om$-word, we write $\sigma(n)=a_n$, $\sigma[n]=\sigma(1)\sigma(2)\ldots \sigma(n)$  for all $n\geq 1$ and $\sigma[0]=\epsilon$.

\medskip
The usual concatenation product of two finite words $u$ and $v$ is
denoted $u\cdot v$ (and sometimes just $uv$). This product is extended to the product of a finite word $u$ and an $\om$-word $v$: the infinite word $u\cdot v$ is then the $\om$-word such that:

\medskip
 $(u\cdot v)(k)=u(k)$  if $k\leq |u|$ , and
 $(u\cdot v)(k)=v(k-|u|)$  if $k>|u|$.

\medskip
The {\it set of } $\om$-{\it words} over  an alphabet $\Si$ is denoted by $\Si^\om$. An  $\om$-{\it language} $V$ over an alphabet $\Sigma$ is a subset of  $\Si^\om$, and its  complement (in $\Sio$) is $\Sio \setminus V$, denoted $V^-$.

The {\it prefix relation} is denoted $\sqsubseteq$: a finite word $u$ is a {\it prefix} \label{def:prefix}
of a finite word $v$ (respectively,  an infinite word $v$), denoted $u\sqsubseteq v$, if and only if there exists a finite word $w$ (respectively, an infinite word $w$), such that $v=u\cdot w$.

\subsection{Counter automata}
\label{ssec:counter-automata}

Let $k$ be an integer with $k\geq 1$. A $k$-counter machine has $k$ {\it counters}, each of which containing a non-negative integer. The machine can test whether the content of a given counter is zero or not, but this is not possible if the counter is a blind (sometimes called partially blind, as in \cite{Gre78}) counter. This means that if a transition of the machine is enabled  when the content of a blind counter is zero then the same transition is also enabled when the content of the same counter is a positive  integer.
The transitions depend on the letter read by the machine, the current state of the finite control, and the tests about the values of the counters.
Notice that in the sequel we shall only consider real-time automata, i.e.~$\epsilon$-transitions are not allowed (but the general results of this paper will be easily extended to the case of non-real-time automata).

Formally, a non-deterministic \emph{real time} $k$-counter machine is a 4-tuple $\mathcal{M}=\langle K,\Si, \Delta, q_0\rangle$,  where $K$ is a finite set of \emph{states}, $\Sigma$ is a finite input alphabet, $q_0\in K$ is an \emph{initial state}, and $\Delta \subseteq K \times \Si\times \{0, 1\}^k \times K \times \{0, 1, -1\}^k$ is a \emph{transition relation}.

If  the machine $\mathcal{M}$ is in a state $q$ and $c_i \in \N$ is the content of the $i^{th}$ counter $\mathcal{C}_i$ for $i=1,\ldots,k$; then the \emph{configuration} (or \emph{global state}) of $\mathcal{M}$ is the $(k{+}1)$-tuple $(q, c_1, \ldots , c_k)$.

\medskip
Consider $a\in \Si$, $q, q' \in K$, and $(c_1, \ldots , c_k) \in \N^k$ such that $c_j=0$ for $j\in E \subseteq  \{1, \ldots , k\}$ and $c_j >0$ for $j\notin E$. If $(q, a, i_1, \ldots , i_k, q', j_1, \ldots , j_k) \in \Delta$ where $i_j=0$ for $j\in E$ and $i_j=1$ for $j\notin E$, then we write:

~~~~~~~~$a: (q, c_1, \ldots , c_k)\mapsto_{\mathcal{M}} (q', c_1+j_1, \ldots , c_k+j_k)$.

\medskip
 Thus  the transition relation must obviously satisfy:
 \nl if $(q, a, i_1, \ldots , i_k, q', j_1, \ldots , j_k)  \in    \Delta$ and  $i_m=0$ for
 some $m\in \{1, \ldots , k\}$  then $j_m=0$ or $j_m=1$ (but $j_m$ may not be equal to $-1$).

Moreover, if the counters of $\mathcal{M}$ are  blind, then, if $(q, a, i_1, \ldots , i_k, q', j_1, \ldots , j_k)  \in    \Delta$  holds, and  $i_m=0$ for some $m\in \{1, \ldots , k\}$ then  $(q, a, i_1, \ldots , i_k, q', j_1, \ldots , j_k)  \in    \Delta$  also holds if $i_m=1$ and the other integers are unchanged.

An $\om$-sequence of configurations $r=(q_i, c_1^{i}, \ldots c_k^{i})_{i \geq 1}$ is called
a run of $\mathcal{M}$ on an $\om$-word  $\sigma=a_1a_2 \ldots a_n \ldots $   over $\Si$ iff:\smallskip

(1)  $(q_1, c_1^{1}, \ldots c_k^{1})=(q_0, 0, \ldots, 0)$

(2)   for each $i\geq 1$,\medskip

~~ $a_i: (q_i, c_1^{i}, \ldots c_k^{i})\mapsto_{\mathcal{M}}
(q_{i+1},  c_1^{i+1}, \ldots c_k^{i+1})$.\medskip

For every such run $r$, $\mathrm{In}(r)$ is the set of all states visited infinitely many times during $r$.

\begin{definition}
\label{def:buchi-k-counter}
A B\"uchi $k$-counter automaton is a 5-tuple $\mathcal{M}=\langle K,\Si,$ $\Delta, q_0, F\rangle$, where $ \mathcal{M}'=\langle K,$ $\Si,$  $\Delta, q_0\rangle$ is a $k$-counter machine and $F \subseteq K$ is a set of \emph{accepting states}. The $\om$-language accepted by~$\mathcal{M}$ is: $\lang(\mathcal{M})\eqdef \{  \sigma\in\Si^\om \mid \text{  there exists a run $r$ of $\mathcal{M}$ on $\sigma$ such that $\mathrm{In}(r)
 \cap F \neq \emptyset$} \}$.
\end{definition}

\begin{definition}
A Muller $k$-counter automaton is a 5-tuple $\mathcal{M}=\langle K,$ $\Si, \Delta, q_0, \mathcal{F}\rangle$, where $\mathcal{M}'=\langle K,$ $\Si, \Delta, q_0\rangle$ is a $k$-counter machine and $\mathcal{F}\subseteq 2^K$ is a set of accepting sets of states. The $\om$-language accepted by $\mathcal{M}$ is: $\lang(\mathcal{M})\eqdef \{  \sigma\in\Si^\om \mid \text{there exists a run $r$ of $\mathcal{M}$ on $\sigma$ such that $\mathrm{In}(r) \in  \mathcal{F}$} \}$.
\end{definition}

Given a~configuration $(q,c_1,\ldots,c_k)$ one can extend the above definitions to $\lang\big(\mathcal{M}, (q,c_1,\ldots,c_k)\big)$ which is the set of $\omega$\=/words accepted by $\mathcal{M}$ starting from the configuration $(q,c_1,\ldots,c_k)$.\label{def:sim-order}

A counter machine $\Mm=\langle K,\Si,$ $\Delta, q_0\rangle$ is \emph{deterministic} if its transition relation $\Delta$ is functional in the following sense: for each $q\in K$, $a\in \Sigma$, and $(i_1,\ldots,i_k)\in \{0,1\}^k$ there is exactly one transition of the machine of the form $(q,a,i_1,\ldots,i_k,q',j_1,\ldots,j_k)$ for some $q'\in K$ and $(j_1,\ldots,j_k)\in\{-1,0,1\}^k$.

Notice that the definition of a deterministic counter machine $\Mm$ ensures that for every input $\omega$\=/word $\sigma=a_1a_2\ldots$ there exists a unique run $\runA$ of $\Mm$ on $\sigma$. This observation motivates the following definition. We call a B\"uchi (resp. Muller) k-counter automaton $\Mm$ \emph{unambiguous} if for every input $\omega$-word $\sigma$ there exists at most one accepting run $\runA$ of $\Mm$ on $\sigma$. Notice that the uniqueness requirement is not limited to the sequence of states visited in $\runA$ but also the exact counter values --- two runs visiting the same states but with different counter values are considered distinct. Similarly, we say that $\Mm$ is \emph{countably unambiguous} if for every input $\omega$-word $\sigma$ there exist at most countably many accepting runs of $\Mm$ on $\sigma$.

The semantic condition of unambiguity is known to be very intriguing~\cite{colcombet_determinism}, with certain tractability features~\cite{stearns_unambiguous,mottet_unambiguous_register}, and expressive power ranging between deterministic and non\=/deterministic models~\cite{carton_michel_prophetic,walukiewicz_choice}. Also, various \emph{degrees of ambiguity},
have been studied, see e.g.~\cite{Fink-Sim,Fin-ambTM,Fin03b,FinkelS09,rabinovich_degrees} and the references therein.

There is a~natural simulation order on the configurations of a~counter automaton: a~configuration $(q,c_1,\ldots,c_k)$ \emph{simulates} $(q,c'_1,\ldots,c'_k)$ (denoted $(q,c_1,\ldots,c_k)\succeq(q,c'_1,\ldots,c'_k)$) if they have the same state $q$ and the counter values $c_i$ and $c'_i$ satisfy coordinate\=/wise $c_j\geq c'_j$ for $j=1,\ldots,k$.

\begin{remark}
\label{rem:residual-order}
If $(q,c_1,\ldots,c_k)\succeq(q,c'_1,\ldots,c'_k)$ are two configurations of a $k$-blind counter machine~$\Mm$ then $\lang\big(\mathcal{M}, (q,c_1,\ldots,c_k)\big)\supseteq\lang\big(\mathcal{M},(q,c'_1,\ldots,c'_k)\big)$.
\end{remark}

\begin{proof}
It is enough to notice that each accepting run of $\mathcal{M}$ from $(q,c'_1,\ldots,c'_k)$ can be lifted to an accepting run from $(q,c_1,\ldots,c_k)$ just by increasing the counter values.
\end{proof}

The above remark relies heavily on the assumption of the blindness of counters. Moreover, it is important that the acceptance condition of the machines is defined purely in terms of the visited states --- the counter values do not intervene.

The above remark implies that, if there is exactly one counter, the maximal size of an~anti\=/chain of the simulation order is bounded by the number of states.

It is well known that an $\om$-language is accepted by a non-deterministic (real time)
B\"uchi $k$-counter automaton iff it is accepted by a non-deterministic
(real time) Muller  $k$-counter automaton \cite{eh}. Notice that it cannot be shown without using the non determinism of automata and this result is no longer true in the  deterministic case.

The class of $\om$-languages accepted by real time $k$-counter B\"uchi  automata  (respectively, real time $k$-blind-counter B\"uchi automata) is denoted {\bf r}-${\bf CL}(k)_\om$ (respectively, {\bf r}-${\bf BCL}(k)_\om$). (Notice that in previous papers, as in \cite{finkel2006borel}, the class {\bf r}-${\bf CL}(k)_\om$ was denoted {\bf r}-${\bf BCL}(k)_\om$ so we have slightly changed the notation in order to distinguish the different classes).

The class ${\bf CL}(1)_\om$ is a strict subclass of the class ${\bf CFL}_\om$ of context free $\om$-languages accepted by pushdown B\"uchi automata.

If we omit the counter of a real-time B\"uchi $1$-counter automaton, then we simply get the notion of B\"uchi automaton. The class of $\om$-languages accepted by  B\"uchi automata is the class of regular $\om$-languages.

\section{Hierarchies in Cantor space}
\label{sec:topology}

\subsection{Borel hierarchy and analytic sets}

 We assume the reader to be familiar with basic notions of topology which
may be found in \cite{Moschovakis80,LescowThomas,Staiger97,PerrinPin}.
There is a natural metric on the set $\Sio$ of  infinite words
over a finite alphabet
$\Si$ containing at least two letters which is called the {\it prefix metric} and is defined as follows. For $u, v \in \Sio$ and
$u\neq v$ let $\delta(u, v)=2^{-l_{\mathrm{pref}(u,v)}}$ where $l_{\mathrm{pref}(u,v)}$
 is the first integer $n$
such that the $(n{+}1)^{st}$ letter of $u$ is different from the $(n{+}1)^{st}$ letter of~$v$.
This metric induces on $\Sio$ the usual  Cantor topology in which the {\it open subsets} of
$\Sio$ are of the form $W\cdot \Si^\om$, for $W\subseteq \Sis$.
A set $L\subseteq \Si^\om$ is a {\it closed set} iff its complement $\Si^\om - L$
is an open set.

 Define now the {\it Borel hierarchy} of subsets of $\Si^\om$:

\begin{definition}
For a non-null countable ordinal $\alpha$, the classes ${\bf \Si}^0_\alpha$
 and ${\bf \Pi}^0_\alpha$ of the Borel hierarchy on the topological space $\Si^\om$
are defined as follows:
\nl ${\bf \Si}^0_1$ is the class of open subsets of $\Si^\om$,
 ${\bf \Pi}^0_1$ is the class of closed subsets of $\Si^\om$,
\nl and for any countable ordinal $\alpha \geq 2$:
\nl ${\bf \Si}^0_\alpha$ is the class of countable unions of subsets of $\Si^\om$ in
$\bigcup_{\gamma <\alpha}{\bf \Pi}^0_\gamma$.
 \nl ${\bf \Pi}^0_\alpha$ is the class of countable intersections of subsets of $\Si^\om$ in
$\bigcup_{\gamma <\alpha}{\bf \Si}^0_\gamma$.
\end{definition}

\noindent
 The class of
{\it Borel\ sets} is $\borel\eqdef\bigcup_{\xi <\omega_1}\ \borapxi\! =\!
\bigcup_{\xi <\omega_1}\ \bormpxi$, where $\om_1$ is the first uncountable ordinal.
There are also some subsets of $\Si^\om$ which are not Borel.
In particular
the class of Borel subsets of $\Si^\om$ is strictly included into
the class  ${\bf \Si}^1_1$ of {\it analytic sets} which are
obtained by projection of Borel sets.

\begin{definition}
\label{def:analytic}
A subset $A$ of  $\Si^\om$ is in the class ${\bf \Si}^1_1$ of {\it analytic} sets
if the following condition is satisfied: there exists another finite set~$Y$ and a Borel subset $B$  of  $(\Si \times Y)^\om$
such that $ x \in A$ iff $\exists y \in Y^\om $ such that $(x, y) \in B$,
where $(x, y)$ is the infinite word over the alphabet $\Si \times Y$ such that
$(x, y)(i)=(x(i),y(i))$ for each  integer $i\geq 1$.
\end{definition}

   We now define completeness with regard to reduction by continuous functions.
For a countable ordinal  $\alpha\geq 1$, a set $F\subseteq \Si^\om$ is said to be
a ${\bf \Si}^0_\alpha$
(respectively,  ${\bf \Pi}^0_\alpha$, ${\bf \Si}^1_1$)-{\it complete set}
iff for any set $E\subseteq Y^\om$  (with $Y$ a finite alphabet):
 $E\in {\bf \Si}^0_\alpha$ (respectively,  $E\in {\bf \Pi}^0_\alpha$,  $E\in {\bf \Si}^1_1$)
iff there exists a continuous function $f: Y^\om \ra \Si^\om$ such that $E = f^{-1}(F)$.

 Let us now  recall the definition of the  arithmetical hierarchy of  $\om$-languages,
see for example \cite{Staiger97,Moschovakis80}.
 Let $\Si$ be a finite alphabet. An $\om$-language $L\subseteq \Si^\om$  belongs to the class
$\Si_n$ iff there exists a recursive relation
$R_L\subseteq (\mathbb{N})^{n-1}\times \Si^\star$  such that
~~$L = \{\sigma \in \Si^\om \mid \exists a_1\ldots Q_na_n  \quad (a_1,\ldots , a_{n-1},
\sigma[a_n+1])\in R_L \}$,
\noindent where $Q_i$ is one of the quantifiers $\forall$ or $\exists$
(not necessarily in an alternating order). An $\om$-language $L\subseteq \Si^\om$  belongs to the class
$\Pi_n$ if and only if its complement $\Si^\om - L$  belongs to the class
$\Si_n$.
The inclusion relations that hold  between the classes $\Si_n$ and $\Pi_n$ are
the same as for the corresponding classes of the Borel hierarchy and
the classes $\Si_n$ and $\Pi_n$ are strictly included in the respective classes
${\bf \Si}_n^0$ and ${\bf \Pi}_n^0$ of the Borel hierarchy.

As in the case of the Borel hierarchy, projections of arithmetical sets
(of the second $\Pi$-class) lead
beyond the arithmetical hierarchy, to the analytical hierarchy of $\om$-languages.
 The first class of the analytical hierarchy of $\om$-languages
 is the (lightface)  class $\Si^1_1$ of effective analytic sets.
An $\om$-language $L\subseteq \Si^\om$  belongs to the class
$\Si_1^1$ if and only if there exists a recursive relation
$R_L\subseteq (\mathbb{N})\times \{0, 1\}^\star \times \Si^\star$  such that:
~~$L = \{\sigma \in \Si^\om \mid \exists \tau (\tau\in \{0, 1\}^\om \wedge \fa n \exists m
 ( (n, \tau[m], \sigma[m]) \in R_L )) \}$.
\noindent Thus an $\om$-language $L\subseteq \Si^\om$  is in the class $\Si_1^1$ iff it is the projection
of an $\om$-language over the alphabet $\{0, 1\} \times \Si$ which is in the class $\Pi_2$.

Kechris, Marker, and Sami proved in \cite{KMS89} that the supremum
of the set of Borel ranks of  (lightface)~$\Pi_1^1$ (so also of lightface $\Si_1^1$)  sets is the ordinal $\gamma_2^1$.
This ordinal is precisely defined in \cite{KMS89}.
It holds that $ \om_1^{\mathrm{CK}} < \gamma_2^1$, where $\om_1^{\mathrm{CK}}$ is the first non-recursive ordinal,  called the Church-Kleene ordinal.

Notice  that it seems still unknown  whether {\it every } non null ordinal $\gamma < \gamma_2^1$ is the Borel rank
of a (lightface) $\Pi_1^1$ (or $\Si_1^1$) set.
On the other hand it is known that
 for every ordinal
$\gamma < \om_1^{\mathrm{CK}}$
there exist some
${\bf \Si}^0_\gamma$-complete and   ${\bf \Pi}^0_\gamma$-complete
sets in the class $\Delta_1^1$.

  Recall that a B\"uchi Turing machine is just a Turing machine working on infinite
inputs with a B\"uchi-like acceptance condition, and
that the class of  $\om$-languages accepted by  B\"uchi Turing machines
is the class $ \Si^1_1$   \cite{CG78b,Staiger97}.

\subsection{Wadge hierarchy}

\noindent We now introduce the Wadge hierarchy, which is a great refinement of the Borel hierarchy defined
via reductions by continuous functions, \cite{Duparc01,Wadge83}.

\begin{definition}[Wadge \cite{Wadge83}] Let $X$, $Y$ be two finite alphabets.
For $L\subseteq X^\om$ and $L'\subseteq Y^\om$, $L$ is said to be \emph{Wadge reducible} to $L'$
($L\leq _W L')$ iff there exists a continuous function $f: X^\om \ra Y^\om$, such that
$L=f^{-1}(L')$.
$\om$\=/languages $L$ and $L'$ are Wadge equivalent iff $L\leq _W L'$ and $L'\leq _W L$.
This will be denoted by $L\equiv_W L'$. Moreover, we shall say that
$L<_W L'$ iff $L\leq _W L'$ but not $L'\leq _W L$.
\nl  A set $L\subseteq X^\om$ is said to be \emph{self dual} iff  $L\equiv_W L^-$, and otherwise it is said to be \emph{non self dual}.
\end{definition}

\noindent
 The relation $\leq _W $  is reflexive and transitive,
 and $\equiv_W $ is an equivalence relation.
 The {\it equivalence classes} of $\equiv_W $ are called {\it Wadge degrees}.  The Wadge hierarchy $WH$ is the class of Borel subsets of a set  $X^\om$, where  $X$ is a finite set,
 equipped with $\leq _W $ and with $\equiv_W $.
\nl  For $L\subseteq X^\om$ and $L'\subseteq Y^\om$, if
$L\leq _W L'$ and $L=f^{-1}(L')$  where $f$ is a continuous
function from $ X^\om$  into $Y^\om$, then $f$ is called a continuous reduction of $L$ to
$L'$. Intuitively it means that $L$ is less complicated than $L'$ because
to check whether $x\in L$ it suffices to check whether $f(x)\in L'$ where $f$
is a continuous function. Hence the Wadge degree of an $\om$-language
is a measure
of its topological complexity.
\nl
Notice  that in the above definition, we consider that a subset $L\subseteq  X^\om$ is given
together with the alphabet $X$.

\noindent We can now define the {\it Wadge class} of a set $L$:

\begin{definition}
Let $L$ be a subset of $X^\om$. The Wadge class of $L$ is :
\[[L]\eqdef \{ L' \mid  L'\subseteq Y^\om \mbox{ for a finite alphabet }Y   \mbox{  and  } L'\leq _W L \}.\]
\end{definition}

\noindent Recall that each {\it Borel class} $\bsigma{\alpha}$ and $\bpi{\alpha}$
is a {\it Wadge class}.
A set $L\subseteq X^\om$ is a $\bsigma{\alpha}$
 (respectively $\bpi{\alpha}$)-{\it complete set} iff for any set
$L'\subseteq Y^\om$, $L'$ is in
$\bsigma{\alpha}$ (respectively $\bpi{\alpha}$) iff $L'\leq _W L . $

  There is a close relationship between Wadge reducibility
 and games which we now introduce.

\begin{definition} Let
$L\subseteq X^\om$ and $L'\subseteq Y^\om$.
The Wadge game  $W(L, L')$ is a game with perfect information between two players,
player 1 who is in charge of $L$ and player 2 who is in charge of $L'$.
Player 1 first writes a letter $a_1\in X$, then player 2 writes a letter
$b_1\in Y$, then player 1 writes a letter $a_2\in  X$, and so on.
 The two players alternatively write letters $a_n$ of $X$ for player 1 and $b_n$ of $Y$
for player 2.
After $\om$ steps, the player 1 has written an $\om$-word $a\in X^\om$ and the player 2
has written an $\om$-word $b\in Y^\om$.
 The player 2 is allowed to skip, even infinitely often, provided he really writes an
$\om$-word in  $\om$ steps. The player 2 wins the play iff [$a\in L \lra b\in L'$], i.e. iff :
\begin{center}
  [($a\in L ~{\rm and} ~ b\in L'$)~ {\rm or} ~
($a\notin L ~{\rm and}~ b\notin L'~{\rm and} ~ b~{\rm is~infinite}  $)].
\end{center}
\end{definition}

\noindent
Recall that a strategy for player 1 is a function
$\sigma :(Y\cup \{s\})^\star\ra X$.
And a strategy for player 2 is a function $f:X^+\ra Y\cup\{ s\}$.
The strategy $\sigma$ is a winning strategy  for player 1 iff he always wins a play when
 he uses the strategy $\sigma$, i.e. when the  $n^{th}$  letter he writes is given
by $a_n=\sigma (b_1\cdots b_{n-1})$, where $b_i$ is the letter written by player 2
at step $i$ and $b_i=s$ if player 2 skips at step $i$. A winning strategy for player 2 is defined in a similar manner.

      Martin's Theorem states that every Gale-Stewart game $G(X)$ (see \cite{Kechris94}),  with $X$ a Borel set,
is determined and this implies the following :

\begin{theorem} [Wadge] Let $L\subseteq X^\om$ and $L'\subseteq Y^\om$ be two Borel sets, where
$X$ and $Y$ are finite  alphabets. Then the Wadge game $W(L, L')$ is determined:
one of the two players has a winning strategy. And $L\leq_W L'$ iff the player 2 has a
winning strategy  in the game $W(L, L')$.
\end{theorem}

\begin{theorem} [Wadge]\label{wh}
Up to the complement and $\equiv _W$, the class of Borel subsets of $X^\om$,
 for  a finite alphabet $X$  having at least two letters, is a well ordered hierarchy.
 There is an ordinal $|WH|$, called the length of the hierarchy, and a map
$d_W^0$ from $WH$ onto $|WH|-\{0\}$, such that for all $L, L' \subseteq X^\om$:\smallskip
\nl $d_W^0 L < d_W^0 L' \lra L<_W L' $  and
\nl $d_W^0 L = d_W^0 L' \lra [ L\equiv_W L' $ or $L\equiv_W L'^-]$.
\end{theorem}

\noindent
 The Wadge hierarchy of Borel sets of {\bf finite rank }
has  length $^1\varepsilon_0$ where $^1\varepsilon_0$
 is the limit of the ordinals $\alpha_n$ defined by $\alpha_1=\om_1$ and
$\alpha_{n+1}=\om_1^{\alpha_n}$ for $n$ a non negative integer, $\om_1$
 being the first non countable ordinal. Then $^1\varepsilon_0$ is the first fixed
point of the ordinal exponentiation of base $\om_1$. The length of the Wadge hierarchy
of Borel sets in $\bdelta{\om}= \bsigma{\om}\cap \bpi{\om}$
  is the $\om_1^{th}$ fixed point
of the ordinal exponentiation of base $\om_1$, which is a much larger ordinal. The length
of the whole Wadge hierarchy of Borel sets is a huge ordinal, with regard
to the $\om_1^{th}$ fixed point
of the ordinal exponentiation of base $\om_1$. It is described in \cite{Wadge83,Duparc01}
by the use of the Veblen functions.

\section{Wadge degrees of $\om$-languages of Petri nets}
\label{sec:wadge}

We are firstly going to prove the following result.

\begin{theorem}\label{thewad}
The Wadge hierarchy of the class {\bf r}-${\bf BCL}(4)_\om$ is equal to the Wadge hierarchy of the class
{\bf r}-${\bf CL}(1)_\om$.
\end{theorem}

 In order to prove this result, we  first define
 a coding of $\om$-words
over a finite alphabet $\Si$ by $\om$-words over the alphabet  $\Si\cup\{A, B, 0\}$ where
$A$, $B$ and $0$  are  new letters not in $\Si$.

\medskip
We shall code an $\om$-word $x\in \Si^{\om}$ by the $\om$-word $h(x)$ defined by
\begin{center}
$h(x)=A0x(1)B0^{2}x(2)A\cdots      B0^{2n}x(2n)A0^{2n+1}x(2n+1)B \cdots $
\end{center}
\noindent   This coding defines a  mapping  $h: \Si^{\om} \ra (\Si\cup\{A, B, 0\})^\om$.
\nl The
function $h$ is continuous because for all $\om$-words $x, y \in \Si^{\om}$ and
each positive integer $n$,  it holds
that  $\delta(x, y) < 2^{-n} \ra \delta( h(x), h(y) ) < 2^{-n}$.
\bigskip

We are going to  state  Lemma \ref{lem5}. Before that, we just describe some important facts given by this lemma and  its proof.  The lemma provides, from a
 real time $1$-counter B\"uchi automaton $\mathcal{A}$ accepting $\om$\=/words over the alphabet $\Si$,   a construction of a  4-blind-counter B\"uchi automaton $\mathcal{B}$
reading $\om$\=/words over the alphabet $\Ga=\Si \cup\{A, B, 0\}$ which is able in some sense to simulate the automaton  $\mathcal{A}$.  Actually the automaton $\mathcal{B}$
simulates the reading of an  $\om$\=/word $x$ by  $\mathcal{A}$ only when $\mathcal{B}$  reads the specific $\om$\=/word
\begin{center}
$h(x)=A0x(1)B0^{2}x(2)A\cdots      B0^{2n}x(2n)A0^{2n+1}x(2n+1)B \cdots $
\end{center}
\eject
\noindent  The reading by the  automaton $\mathcal{B}$ of the $\om$\=/word $h(x)$ will provide a  decomposition of the $\om$-word $h(x)$ of the following form:\vspace*{-1mm}
\begin{center}
$y=Au_1v_1x(1)Bu_2v_2x(2)Au_3v_3x(3)B \cdots  $

 ~~~~~~~~~~~~~~~~~~~~~~$ \cdots  Bu_{2n}v_{2n}x(2n)Au_{2n+1}v_{2n+1}x(2n+1)B \cdots$
 \end{center}
\noindent where,  for all integers $i\geq 1$, $ u_i, v_i  \in 0^\star$, $x(i) \in \Si$, $|u_1|=0$.
Then an accepting run of $\mathcal{B}$ on  $h(x)$ will correspond to an accepting run of  $\mathcal{A}$ on $x$. Moreover the successive values of the single counter of $\mathcal{A}$
during this run will be the integers $|u_n|$, $n\geq 1$. Then the automaton $\mathcal{B}$ will be able to determine, using a finite control component during the reading of the finite word $u_n$, whether $|u_n|=0$, and thus to simulate the zero-tests of the automaton $\mathcal{A}$.

The proof of  the following lemma will explain this in detail.

\begin{lemma}\label{lem5}  Let $\mathcal{A}$ be a real time
$1$-counter B\"uchi automaton accepting
$\om$\=/words over the alphabet $\Si$. Then one can construct a real time 4-blind-counter B\"uchi automaton $\mathcal{B}$
reading $\om$\=/words over the alphabet $\Ga=\Si \cup\{A, B, 0\}$, such that
$ \lang(\mathcal{A})$ = $h^{-1} (\lang(\mathcal{B}))$, i.e.
 $\fa x \in \Si^{\om}.\ h(x) \in \lang(\mathcal{B}) \longleftrightarrow
x\in  \lang(\mathcal{A}).$
\end{lemma}

\begin{proof}
Let $\mathcal{A}=(K,\Si, \Delta, q_0, F)$ be a real time
$1$-counter B\"uchi automaton accepting
$\om$-words over the alphabet $\Si$. We are going to explain informally the behaviour of the 4-blind-counter B\"uchi automaton $\mathcal{B}$
when reading an $\om$-word of the form $h(x)$, even if we are going to  see  that $\mathcal{B}$ may also accept some infinite words which do not belong to the
range of $h$. Recall that $h(x)$ is of the form
\begin{center}
$h(x)=A0x(1)B0^{2}x(2)A \cdots
B0^{2n}x(2n)A0^{2n+1}x(2n+1)B \cdots $
\end{center}
\noindent Notice that in particular every $\om$-word in $h(\Sio)$ is of the form:
\begin{center}
$ y = A0^{n_1}x(1)B0^{n_2}x(2)A \cdots
 B0^{n_{2n}}x(2n)A0^{n_{2n+1}}x(2n+1)B \cdots $
 \end{center}

\noindent where for all $i\geq 1$,  $n_i >0$ is a  positive integer, and $x(i)\in \Si$.

\medskip
 Moreover it is easy to see that the set of $\om$-words $y \in \Gao$  which can be written in the  above form is a regular $\om$-language $\mathcal{R}\subseteq \Gao$, and thus we can assume, using a classical product construction (see for instance \cite{PerrinPin}),  that the automaton
$\mathcal{B}$ will only  accept some $\om$-words of this form.

\medskip
 Now the  reading by the automaton $\mathcal{B}$  of  an $\om$-word of  the above form
\begin{center}
$ y = A0^{n_1}x(1)B0^{n_2}x(2)A \cdots   B0^{n_{2n}}x(2n)A0^{n_{2n+1}}x(2n+1)B \cdots $
\end{center}
  \noindent  will give a  decomposition of the $\om$-word $y$ of the following form:
\begin{center}
$y=Au_1v_1x(1)Bu_2v_2x(2)Au_3v_3x(3)B \cdots  $

 ~~~~~~~~~~~~~~~~~~~~~~$ \cdots  Bu_{2n}v_{2n}x(2n)Au_{2n+1}v_{2n+1}x(2n+1)B \cdots$
 \end{center}
\noindent where,  for all integers $i\geq 1$, $ u_i, v_i  \in 0^\star$, $x(i) \in \Si$, $|u_1|=0$.

\eject
 The automaton  $\mathcal{B}$  will  use its four {\it blind} counters, which we denote $\mathcal{C}_1,  \mathcal{C}_2,  \mathcal{C}_3,  \mathcal{C}_4,$  in the following way.
Recall that the automaton  $\mathcal{B}$  being non-deterministic, we do not describe the unique run of $\mathcal{B}$ on $y$, but the general case of a possible run.

At the beginning of the run, the value of each of the four counters is equal to zero. Then  the counter  $\mathcal{C}_1$ is increased of $|u_1|$ when reading $u_1$, i.e. the counter  $\mathcal{C}_1$ is actually not increased since $|u_1|=0$ and the finite control is here used to check this.
Then the counter $\mathcal{C}_2$ is increased of $1$ for each letter $0$ of $v_1$ which is read until the automaton  reads the letter $x(1)$ and then the letter $B$. Notice that at this time the values of the counters $\mathcal{C}_3$ and   $\mathcal{C}_4$ are still equal to zero.
Then the behaviour of the automaton $\mathcal{B}$ when reading the next segment $0^{n_2}x(2)A$ is as follows. The counter
 $\mathcal{C}_1$ is firstly decreased of $1$ for each letter $0$ read,  when reading $k_2$ letters $0$, where $k_2\geq 0$ (notice that here $k_2=0$ because the value of the counter
$\mathcal{C}_1$ being equal to zero, it cannot decrease under $0$).  Then   the counter $\mathcal{C}_2$ is decreased of $1$ for each letter $0$ read, and next the automaton has to read one more letter $0$, leaving unchanged the counters $\mathcal{C}_1$ and   $\mathcal{C}_2$, before reading the letter $x(2)$.  The end of the decreasing mode of $\mathcal{C}_1$ coincide with the
beginning of the decreasing mode of $\mathcal{C}_2$, and this change may occur in a {\it non-deterministic way} (because the automaton $\mathcal{B}$ cannot check whether the value of $\mathcal{C}_1$ is equal to zero).
 Now we describe the behaviour of the counters $\mathcal{C}_3$ and   $\mathcal{C}_4$ when reading the segment $0^{n_2}x(2)A$. Using its finite control,  the automaton $\mathcal{B}$
 has checked that $|u_1|=0$, and then  if there is a transition of
the automaton  $\mathcal{A}$ such that $x(1) : ( q_{0}, |u_1|) \mapsto_{\mathcal{A}}
(q_1, |u_1| +N_1 )$ then the counter $\mathcal{C}_3$ is increased of $1$ for each letter $0$ read, during the reading of the $k_2+N_1$ first  letters $0$ of $0^{n_2}$, where $k_2$ is described above as the number of which the counter $\mathcal{C}_1$ has been decreased. This determines $u_2$ by $|u_2|=k_2+N_1$ and then the counter $\mathcal{C}_4$ is increased by $1$ for each letter $0$ read until $\mathcal{B}$ reads $x(2)$, and this determines $v_2$.  Notice that the automaton $\mathcal{B}$ keeps in its finite control the memory of the state $q_1$ of the automaton $\mathcal{A}$, and  that, after having read the segment $0^{n_2}=u_2v_2$, the values of the counters $\mathcal{C}_3$ and   $\mathcal{C}_4$ are respectively
$|\mathcal{C}_3|=|u_2|=k_2+N_1$ and $|\mathcal{C}_4|=|v_2|=n_2-(|u_2|)$.

Now the run will continue. Notice that generally when reading a segment $B0^{n_{2n}}x(2n)A$  the counters $\mathcal{C}_1$ and $\mathcal{C}_2$ will  successively decrease when reading
the first $(n_{2n}-1)$ letters $0$ and then will remain unchanged when reading the last letter $0$, and the counters $\mathcal{C}_3$ and   $\mathcal{C}_4$ will successively increase, when reading the $(n_{2n})$ letters $0$. Again the end of the decreasing mode of $\mathcal{C}_1$ coincide with the beginning of the decreasing mode of $\mathcal{C}_2$, and this change may occur in a {\it non-deterministic way}.  But the automaton has kept in its finite control
whether $|u_{2n-1}|=0$ or not and also a state $q_{2n-2}$ of the automaton  $\mathcal{A}$. Now,  if  there is a transition of
the automaton  $\mathcal{A}$ such that $x(2n-1) : ( q_{2n-2}, |u_{2n-1}|) \mapsto_{\mathcal{A}}
(q_{2n-1}, |u_{2n-1}| +N_{2n-1} )$ for some integer $N_{2n-1} \in \{-1; 0, 1\}$, and the counter     $\mathcal{C}_1$ is decreased of $1$ for each letter $0$ read,  when reading $k_{2n}$ first letters $0$ of $0^{n_{2n}}$, then the counter  $\mathcal{C}_3$ is increased of $1$ for each letter $0$ read, during the reading of the $k_{2n}+N_{2n-1}$ first  letters $0$ of $0^{n_{2n}}$, and next the   counter $\mathcal{C}_4$  is increased by $1$ for each letter $0$ read until $\mathcal{B}$ reads $x(2n)$, and this determines $v_{2n}$.  Then after having read the segment $0^{n_ {2n}}=u_{2n}v_{2n}$, the values of the counters $\mathcal{C}_3$ and   $\mathcal{C}_4$ have respectively increased of $|u_{2n}|=k_{2n}+N_{2n-1}$ and $|v_{2n}|=n_{2n}-|u_{2n}|$.
Notice that one cannot ensure that, after the reading of $0^{n_ {2n}}=u_{2n}v_{2n}$,  the exact values of these counters  are   $|\mathcal{C}_3|=|u_{2n}|=k_{2n}+N_{2n-1}$ and $|\mathcal{C}_4|=|v_{2n}| =n_ {2n} - |u_{2n}|$. Actually this is due to the fact that one cannot ensure that the values of $\mathcal{C}_3$ and   $\mathcal{C}_4$ are equal to zero at the beginning of the reading of the segment $B0^{n_{2n}}x(2n)A$ although we will see this is true and important in the particular case of a word of the form $y=h(x)$.

The run will continue in a similar manner during the reading of the next segment 	$A0^{n_{2n+1}}x(2n+1)B$, but here the role of the counters $\mathcal{C}_1$ and $\mathcal{C}_2$ on one side, and of the counters $\mathcal{C}_3$ and   $\mathcal{C}_4$  on the other side, will be interchanged.  More precisely the counters $\mathcal{C}_3$ and $\mathcal{C}_4$ will  successively decrease when reading
the first $(n_{2n+1}-1)$ letters $0$ and then will remain unchanged when reading the last letter $0$, and the counters $\mathcal{C}_1$ and   $\mathcal{C}_2$ will successively increase, when reading the $(n_{2n+1})$ letters $0$. The end of the decreasing mode of $\mathcal{C}_3$ coincide with the beginning of the decreasing mode of $\mathcal{C}_4$, and this change may occur in a {\it non-deterministic way}.  But the automaton has kept in its finite control
whether $|u_{2n}|=0$ or not and also a state $q_{2n-1}$ of the automaton  $\mathcal{A}$. Now,  if  there is a transition of
the automaton  $\mathcal{A}$ such that $x(2n) : ( q_{2n-1}, |u_{2n}|) \mapsto_{\mathcal{A}}
(q_{2n}, |u_{2n}| +N_{2n} )$ for some integer $N_{2n} \in \{-1; 0, 1\}$, and the counter     $\mathcal{C}_3$ is decreased of $1$ for each letter $0$ read,  when reading $k_{2n+1}$ first letters $0$ of $0^{n_{2n+1}}$, then the counter  $\mathcal{C}_1$ is increased of $1$ for each letter $0$ read, during the reading of the $k_{2n+1}+N_{2n}$ first  letters $0$ of $0^{n_{2n+1}}$, and next the   counter $\mathcal{C}_2$  is increased by $1$ for each letter $0$ read until $\mathcal{B}$ reads $x(2n+1)$, and this determines $v_{2n+1}$.  Then after having read the segment $0^{n_ {2n+1}}=u_{2n+1}v_{2n+1}$, the values of the counters $\mathcal{C}_1$ and   $\mathcal{C}_2$ have respectively increased of $|u_{2n+1}|=k_{2n+1}+N_{2n}$ and $|v_{2n+1}|=n_{2n+1}-|u_{2n+1}|$.
Notice that again one cannot ensure that, after the reading of $0^{n_ {2n+1}}=u_{2n+1}v_{2n+1}$,  the exact values of these counters  are   $|\mathcal{C}_1|=|u_{2n+1}|=k_{2n+1}+N_{2n}$ and $|\mathcal{C}_2|=|v_{2n+1}| =n_ {2n+1} - |u_{2n+1}|$. This is due to the fact that one cannot ensure that the values of $\mathcal{C}_1$ and   $\mathcal{C}_2$ are equal to zero at the beginning of the reading of the segment $A0^{n_{2n+1}}x(2n+1)B$ although we will see this is true and important in the particular case of a word of the form $y=h(x)$.

The run then continues in the same way if it is possible and in particular if there is no blocking due to the fact that one of the counters of  the automaton  $\mathcal{B}$ would have a negative value.

Now an $\om$-word $y\in \mathcal{R}\subseteq \Gao$ of the above form will be accepted by the automaton  $\mathcal{B}$ if there is such an infinite  run for which a final state $q_f\in F$
of the automaton  $\mathcal{A}$ has been stored infinitely often in the finite control of $\mathcal{B}$ in the way which has just been described above.

 We now consider the particular case of an $\om$-word of the form $y=h(x)$,  for some  $x\in \Sio$.  Let then

 $y=h(x)=A0x(1)B0^{2}x(2)A0^{3}x(3)B \cdots B0^{2n}x(2n)A0^{2n+1}x(2n+1)B \cdots $

\medskip
We are going to show that, if  $y$ is accepted by the automaton $\mathcal{B}$, then $x\in \lang(\mathcal{A})$. Let us consider a run of the automaton $\mathcal{B}$ on $y$ as described above and which is an accepting run. We first show by induction on $n\geq 1$, that after having read an initial segment
of the form
\[A0x(1)B0^{2}x(2)A\cdots A0^{2n-1}x(2n-1)B,\]
the values of the counters $\mathcal{C}_3$ and   $\mathcal{C}_4$ are equal to zero, and the values of the
counters $\mathcal{C}_1$ and   $\mathcal{C}_2$ satisfy $|\mathcal{C}_1| + |\mathcal{C}_2|=2n-1$. And similarly after having read an initial segment
of the form
\[A0x(1)B0^{2}x(2)A \cdots B0^{2n}x(2n)A,\]
the values of the counters $\mathcal{C}_1$ and   $\mathcal{C}_2$ are equal to zero, and the values of the
counters $\mathcal{C}_3$ and   $\mathcal{C}_4$ satisfy $|\mathcal{C}_3| + |\mathcal{C}_4|=2n$.

\medskip
For $n=1$, we have seen that after having read the initial segment $A0x(1)B$, the values of the counters $\mathcal{C}_1$ and   $\mathcal{C}_2$ will be respectively $0$ and $|v_1|$ and here $|v_1|=1$ and thus $|\mathcal{C}_1| + |\mathcal{C}_2|=1$. On the other hand the counters  $\mathcal{C}_3$ and   $\mathcal{C}_4$  have not yet increased so that the value of each of these counters is equal to zero. During  the reading of the segment $0^{2}$ of  $0^{2}x(2)A$ the counters $\mathcal{C}_1$ and   $\mathcal{C}_2$  successively decrease. But here
 $\mathcal{C}_1$ cannot decrease  (with the above notations, it holds that $k_2=0$) so $\mathcal{C}_2$ must decrease of $1$   because after the decreasing mode the automaton $\mathcal{B}$ must read a last letter $0$ without decreasing
 the counters  $\mathcal{C}_1$ and   $\mathcal{C}_2$ and then the letter $x(2)\in \Si$. Thus after having read $0^{2}x(2)A$ the values of $\mathcal{C}_1$ and   $\mathcal{C}_2$ are equal to zero. Moreover the counters $\mathcal{C}_3$ and   $\mathcal{C}_4$ had their values equal to zero at the beginning of the reading of $0^{2}x(2)A$ and they successively increase during the reading of $0^{2}$ and they remain unchanged during the reading of $x(2)A$ so that their values satisfy  $|\mathcal{C}_3| + |\mathcal{C}_4|=2$ after the reading of
  $0^{2}x(2)A$.

  Assume now that for some integer $n>1$  the claim is proved for all integers $k<n$ and let us prove it for the integer $n$. By induction hypothesis we know that at the beginning of the reading of the segment
  $A0^{2n-1}x(2n-1)B$ of $y$, the  values of the counters $\mathcal{C}_1$ and   $\mathcal{C}_2$ are equal to zero, and the values of the
counters $\mathcal{C}_3$ and   $\mathcal{C}_4$ satisfy $|\mathcal{C}_3| + |\mathcal{C}_4|=2n-2$. When reading the $(2n-2)$ first letters $0$ of $A0^{2n-1}x(2n-1)B$ the
counters $\mathcal{C}_3$ and   $\mathcal{C}_4$ successively decrease and they must decrease completely because after there must remain only one letter $0$ to be read by
$\mathcal{B}$ before the letter $x(2n-1)$. Therefore after the reading of  $A0^{2n-1}x(2n-1)B$  the values of the counters $\mathcal{C}_3$ and   $\mathcal{C}_4$ are equal to zero.
And since the   values of the counters $\mathcal{C}_1$ and   $\mathcal{C}_2$ are equal to zero before the reading of $0^{2n-1}x(2n-1)B$ and these counters  successively increase during the reading of  $0^{2n-1}$, their values satisfy $|\mathcal{C}_1| + |\mathcal{C}_2|=2n-1$ after the reading of $A0^{2n-1}x(2n-1)B$.
We can reason in a very similar manner for the reading of the next segment  $B0^{2n}x(2n)A$, the role  of the counters $\mathcal{C}_1$ and $\mathcal{C}_2$ on one side, and of the counters $\mathcal{C}_3$ and   $\mathcal{C}_4$  on the other side, being simply  interchanged.  This  ends the proof of the claim by induction on $n$.

  It is now easy to see by induction that for each integer $n\geq 2$, it holds that $k_n=|u_{n-1}|$.  Then, since with the above notations we have
  $|u_{n+1}|=k_{n+1} + N_{n}=|u_{n}| + N_n$, and there is a transition of
the automaton  $\mathcal{A}$ such that $x(n) : ( q_{n-1}, |u_{n}|) \mapsto_{\mathcal{A}}
(q_{n}, |u_{n}| +N_{n})$ for $N_{n} \in \{-1; 0, 1\}$, it holds that  $x(n) : ( q_{n-1}, |u_{n}|) \mapsto_{\mathcal{A}}
(q_{n}, |u_{n+1}| )$. Therefore the sequence $(q_i, |u_{i}|)_{i\geq 0}$ is an accepting run of the automaton $\mathcal{A}$ on the $\om$-word $x$ and $x\in \lang(\mathcal{A})$.
 Notice that the state $q_0$ of the sequence $(q_i)_{i\geq 0}$  is also the initial state
of $\mathcal{A}$.

  Conversely, it is easy to see that if $x\in \lang(\mathcal{A})$ then there exists an accepting run of the automaton $\mathcal{B}$ on the $\om$-word $h(x)$ and $h(x)\in \lang(\mathcal{B})$.
\end{proof}

The above Lemma \ref{lem5} shows that, given  a real time
$1$-counter (with zero-test) B\"uchi automaton $\mathcal{A}$ accepting
$\om$-words over the alphabet $\Si$,  one can construct a real time 4-blind-counter B\"uchi automaton $\mathcal{B}$ which can  simulate the $1$-counter automaton  $\mathcal{A}$  on the code $h(x)$ of the word $x$. On the other hand, we cannot describe precisely the $\om$-words which are accepted by $\mathcal{B}$ but are not in the set  $h(\Sio)$. However we can see
that all these words have a special shape, as stated by the following lemma.

\begin{lemma}\label{lem6}  Let $\mathcal{A}$ be a real time
$1$-counter B\"uchi automaton accepting
$\om$-words over the alphabet $\Si$,  and let    $\mathcal{B}$   be the real time 4-blind-counter B\"uchi automaton
reading words over the alphabet $\Ga=\Si \cup\{A, B, 0\}$ which is constructed in the proof of Lemma \ref{lem5}.
Let $y \in \lang(\mathcal{B})\setminus h(\Sio)$ being of the following form
\begin{center}
$y = A0^{n_1}x(1)B0^{n_2}x(2)A0^{n_3}x(3)B \cdots  B0^{n_{2n}}x(2n)A0^{n_{2n+1}}x(2n+1)B \cdots $
\end{center}
\noindent and let $i_0$ be the smallest integer $i$ such that $n_i \neq i$. Then it holds that either $i_0=1$ or $n_{i_0} < i_0$.
\end{lemma}

\begin{proof}
Assume first that  $y \in \lang(\mathcal{B})\setminus h(\Sio)$   is of  the following form\medskip
\nl
$ ~~~~~~~~~~ y = A0^{n_1}x(1)B0^{n_2}x(2)A \cdots   B0^{n_{2n}}x(2n)A0^{n_{2n+1}}x(2n+1)B \cdots $

\medskip \noindent and that the smallest integer $i$ such that $n_i \neq i$ is an even integer $i_0>1$.
Consider an infinite accepting run of $\mathcal{B}$ on $y$. It follows from the  proof of the above Lemma \ref{lem5} that after the reading of the initial segment
\begin{center}
$A0^{n_1}x(1)B0^{n_2}x(2)A \cdots
A0^{{i_0-1}}x(i_0-1)B$
\end{center}
\noindent the values of the counters  $\mathcal{C}_3$ and   $\mathcal{C}_4$ are equal to zero, and the values of the counters $\mathcal{C}_1$ and   $\mathcal{C}_2$ satisfy
$|\mathcal{C}_1| + |\mathcal{C}_2|=i_0-1$. Thus since the two counters must successively decrease during the next $n_ {i_0}-1$  letters $0$, it holds that
$n_ {i_0}-1 \leq i_0-1$ because otherwise either $\mathcal{C}_1$ or  $\mathcal{C}_2$ would block. Therefore $n_{i_0} < i_0$ since $n_{i_0} \neq  i_0$ by definition of $i_0$.
The reasoning is very similar in the case of an odd integer $i_0$,  the role  of the counters $\mathcal{C}_1$ and $\mathcal{C}_2$ on one side, and of the counters $\mathcal{C}_3$ and   $\mathcal{C}_4$  on the other side, being simply  interchanged.
\end{proof}

Let $\mathcal{L} \subseteq \Gao$ be the $\om$-language containing the $\om$-words over $\Gamma$   which belong to one of the following $\om$-languages.
\begin{itemize}
\itemsep=0.9pt
\item $\mathcal{L}$$_1$ is the set of $\om$-words over the alphabet $\Si\cup\{A, B, 0\}$
which have not any initial segment in  $A\cdot 0\cdot \Si \cdot  B$.

\item $\mathcal{L}$$_2$ is the set of $\om$-words over the alphabet $\Si\cup\{A, B, 0\}$
which contain a segment of the form  $B\cdot 0^n \cdot a \cdot A\cdot 0^m\cdot b$ or of the form
$A \cdot 0^n \cdot a \cdot B \cdot 0^m\cdot b$
for some letters $a, b \in \Si$ and some positive integers $m \leq  n$.
\end{itemize}

\begin{lemma}\label{lem7}
The $\om$-language $\mathcal{L}$ is accepted by a (non-deterministic)  real-time $1$-blind counter B\"uchi  automaton.
\end{lemma}

\begin{proof}
First, it is easy to see that
$\mathcal{L}$$_1$ is in fact a regular $\om$-language, and thus it is also accepted by  a real-time $1$-blind counter B\"uchi  automaton (even without active counter).  On the other hand it is also easy to construct a real time
$1$-blind counter B\"uchi  automaton accepting the $\om$-language $\mathcal{L}$$_2$. The class of $\om$-languages accepted by {\it non-deterministic} real time
$1$-blind counter B\"uchi  automata being closed under finite union in an effective way, one can  construct a real time $1$-blind counter B\"uchi  automaton accepting $\mathcal{L}$.
\end{proof}

\begin{lemma}\label{lem8}
 Let $\mathcal{A}$ be a real time
$1$-counter B\"uchi automaton accepting
$\om$-words over the alphabet $\Si$. Then one can construct a  real time $4$-blind counter B\"uchi  automaton $\mathcal{P}_{\mathcal{A}}$ such that ~~
$\lang(\mathcal{P}_{\mathcal{A}}) = h( \lang(\mathcal{A}) ) \cup \mathcal{L}.$
\end{lemma}

\begin{proof}
 Let $\mathcal{A}$ be a real time
$1$-counter B\"uchi automaton accepting
$\om$-words over $\Si$. We have seen in the proof of Lemma \ref{lem5} that one can construct a real time $4$-blind counter B\"uchi  automaton
$\mathcal{B}$
reading words over the alphabet $\Ga=\Si \cup\{A, B, 0\}$, such that
$ \lang(\mathcal{A})$ = $h^{-1} (\lang(\mathcal{B}))$, i.e.
$\fa x \in \Si^{\om} ~~~~ h(x) \in \lang(\mathcal{B}) $$\longleftrightarrow
x\in  \lang(\mathcal{A})$. Moreover by Lemma \ref{lem6} it holds that $\lang(\mathcal{B})\setminus h(\Sio) \subseteq \mathcal{L}$ and thus
~~$h( \lang(\mathcal{A}) ) \cup \mathcal{L} = \lang(\mathcal{B})  \cup \mathcal{L}$.
But by Lemma \ref{lem7} the $\om$-language $\mathcal{L}$ is accepted by a (non-deterministic)  real-time $1$-blind counter B\"uchi  automaton, hence also by a
 real-time $4$-blind counter B\"uchi  automaton. The class of $\om$-languages  accepted by  (non-deterministic)  real-time $4$-blind counter B\"uchi  automata is closed under
 finite union in an effective way, and thus  one can construct a  real time $4$-blind counter B\"uchi  automaton $\mathcal{P}_{\mathcal{A}}$ such that ~ ~
$\lang(\mathcal{P}_{\mathcal{A}}) = h( \lang(\mathcal{A}) ) \cup \mathcal{L}. $
\end{proof}

We are now going to  prove that
if $\lang(\mathcal{A})$$\subseteq \Sio$ is accepted by a
real time $1$-counter
automaton $\mathcal{A}$ with a B\"uchi acceptance condition then
$\lang(\mathcal{P}_{\mathcal{A}}) = h( \lang(\mathcal{A}) )$$ \cup \mathcal{L}$
 will have the same Wadge degree as the $\om$-language
$\lang(\mathcal{A})$, except for some very simple cases.

 We first notice that $h(\Si^{\om})$ is a closed subset of $\Gao$. Indeed  it is the image of the compact set $\Sio$ by the continuous function $h$, and thus it is a compact hence also closed subset of $\Gao = (\Si\cup\{A, B, 0\})^\om$.  Thus its complement
$h(\Si^{\om})^-=(\Si\cup\{A, B, 0\})^\om - h(\Si^{\om})$ is an open subset of $\Gao$. Moreover  the set $\mathcal{L}$ is an open subset of  $\Gao$, as it can be easily seen from its definition and one can easily define, from the definition of the $\om$-language  $\mathcal{L}$, a finitary  language $V \subseteq \Gas$   such that  $\mathcal{L}=V\cdot \Gao$.
We shall also denote $\mathcal{L}'=h(\Si^{\om})^-\setminus \mathcal{L}$ so that $\Gao$ is the disjoint union $\Gao = h(\Si^{\om}) \cup \mathcal{L}  \cup \mathcal{L}'$. Notice that
$\mathcal{L}'$ is the difference of the  two open sets $h(\Si^{\om})^-$ and $\mathcal{L}$.

 We now wish to return to the proof of the above Theorem \ref{thewad} stating that
the Wadge hierarchy of the class {\bf r}-${\bf BCL}(4)_\om$  is equal to the Wadge hierarchy of the class   {\bf r}-${\bf CL}(1)_\om$.

    To prove this result we firstly consider non self dual Borel sets. We recall the definition of Wadge degrees
introduced by Duparc in \cite{Duparc01} and which is a slight modification of the previous one.

\begin{definition}
\begin{enumerate}
\item[(a)] $d_w(\emptyset)=d_w(\emptyset^-)=1$
\item[(b)]  $d_w(L)=\sup \{d_w(L')+1 ~\mid ~~L' {\rm ~non~ self ~dual~ and~}
L'<_W L \} $
\nl (for either $L$ self dual or not, $L>_W \emptyset).$
\end{enumerate}
\end{definition}

\noindent  Wadge and Duparc used  the operation of sum of
sets of infinite words which has as
counterpart the ordinal
addition  over Wadge degrees.

\begin{definition}[Wadge, see \cite{Wadge83,Duparc01}]
Assume that $X\subseteq Y$ are two finite alphabets,
  $Y-X$ containing at least two elements, and that
$\{X_+, X_-\}$ is a partition of $Y-X$ in two non empty sets.
 Let $L \subseteq X^{\om}$ and $L' \subseteq Y^{\om}$, then
 $L' + L \eqdef L\cup \{ u\cdot a\cdot \beta  ~\mid  ~ u\in X^\star , ~(a\in X_+
~and ~\beta \in L' )~
 or ~(a\in X_- ~and ~\beta \in L'^- )\}$
\end{definition}

\noindent This operation is closely related to the {\it ordinal sum}
 as it is stated in the following:

\begin{theorem}[Wadge, see \cite{Wadge83,Duparc01}]\label{thesum}
Let $X\subseteq Y$, $Y-X$ containing at least two elements,
   $L \subseteq X^{\om}$ and $L' \subseteq Y^{\om}$ be
non self dual  Borel sets.
Then $(L+L')$ is a non self dual Borel set and
$d_w( L'+L )= d_w( L' ) + d_w( L )$.
\end{theorem}

\noindent A player in charge of a set $L'+L$ in a Wadge game is like a player in charge of the set $L$ but who
can, at any step of the play,    erase  his previous play and choose to be this time in charge of  $L'$ or of $L'^-$.
Notice that he can do this only one time during a play.

The following lemma was proved in \cite{finkel2006borel}. Notice that below the empty set is considered as an $\om$-language over an alphabet
$\Delta$ such that $\Delta - \Si$ contains at least two elements.

\begin{lemma}\label{sum-wad}
Let $L \subseteq \Sio$ be a non self dual  Borel set such that $d_w( L )\geq \om$. Then it holds that $L \equiv_W \emptyset + L$.
\end{lemma}

We can now prove the following lemma.

\begin{lemma}\label{nad}
Let  $L \subseteq \Sio$ be a non self dual  Borel set accepted by a real time
$1$-counter B\"uchi automaton  $\mathcal{A}$.
Then there is an $\om$-language $L'$ accepted by   a real time $4$-blind counter B\"uchi  automaton such that $L \equiv_W L'$.
\end{lemma}

\begin{proof}
Recall first  that there are regular $\om$-languages of every finite Wadge degree, \cite{Staiger97,Selivanov98}. These regular $\om$-languages
are Boolean combinations of open sets, and they obviously  belong to the class  {\bf r}-${\bf BCL}(4)_\om$  since every regular $\om$-language belongs to this class.

So we have only to consider the case of non self dual Borel sets of Wadge
degrees greater than or equal to $\om$.

\medskip
 Let then $L=\lang(\mathcal{A}) \subseteq \Sio$ be a non self dual  Borel set, accepted by a real time
$1$-counter B\"uchi automaton  $\mathcal{A}$, such that $d_w( L )\geq \om$.
By Lemma \ref{lem8},
$\lang(\mathcal{P}_{\mathcal{A}}) = h( \lang(\mathcal{A}) )$$ \cup \mathcal{L}$ is accepted by a   a real time $4$-blind counter B\"uchi  automaton $\mathcal{P}_{\mathcal{A}}$,
where the mapping $h: \Sio \ra (\Si \cup\{A, B, 0\})^\om$ is defined, for $x\in \Sio$,  by:
\begin{center}
$h(x)= A0x(1)B0^{2}x(2)A0^{3}x(3)B \cdots    B0^{2n}x(2n)A0^{2n+1}x(2n+1)B \cdots $
\end{center}

We set $L'=\lang(\mathcal{P}_{\mathcal{A}})$ and we now prove that  $L' \equiv_W L$.

Firstly, it is easy to see that the function $h$ is a continuous  reduction of $L$ to $L'$ and thus $L \leq_W L'$.

To prove that $L' \leq_W L $, it suffices to prove that $L' \leq_W \emptyset +(\emptyset + L)$ because
Lemma \ref{sum-wad} states that $ \emptyset + L \equiv_W  L$, and thus also $\emptyset + (\emptyset + L) \equiv_W  L$. Consider the Wadge game $W( L',  \emptyset + (\emptyset + L) )$.
Player 2 has a winning strategy in this game which we now describe.

As long as Player 1 remains in the closed set $h(\Si^{\om})$ (this means that the word written by Player 1 is a prefix of some infinite word in $h(\Si^{\om})$)
Player 2 essentially  copies  the play of player 1 except that Player 2 skips when player 1 writes a letter not in $\Si$.
He continues forever with this strategy if the word written by player 1 is always a prefix of some $\om$-word of $h(\Sio)$. Then after $\om$ steps
Player 1 has written an $\om$-word $h(x)$ for some $x \in \Sio$, and Player 2 has written $x$. So in that case
$h(x) \in L'$ iff  $x \in \lang(\mathcal{A})$ iff  $x \in \emptyset + (\emptyset + L)$.

But if at some step of the play, Player 1 ``goes out of" the closed set  $h(\Sio)$ because the word he has now
written is not a prefix of any $\om$-word of  $h(\Sio)$,  then Player 1 ``enters'' in the open set $h(\Sio)^- = \mathcal{L}\cup \mathcal{L}'$ and will stay in this set. Two cases may now appear.

{\bf First case.} When  Player 1 ``enters'' in the open set $h(\Sio)^- = \mathcal{L}\cup \mathcal{L}'$, he actually enters in the open set $ \mathcal{L}=V\cdot \Gao$ (this means that Player 1 has written an initial segment in $V$). Then the final word written by Player 1 will surely be inside $L'$. Player 2 can now  write a letter of $\Delta -\Si$ in such a way that he is now like a player in charge of the whole set  and he can
now writes an $\om$-word $u$ so that his final $\om$-word will be inside $\emptyset + L $, and also inside $\emptyset + (\emptyset + L)$. Thus Player 2 wins this play too.

{\bf Second case.}   When  Player 1 ``enters'' in the open set $h(\Sio)^- = \mathcal{L}\cup \mathcal{L}'$, he does not enter in the open set $ \mathcal{L}=V\cdot \Gao$.
Then  Player 2,  being first like a player in charge of the set $(\emptyset + L)$,  can write a  letter of $\Delta -\Si$ in such a way that he is now like a player in charge of the empty set  and he can
now continue, writing an $\om$-word $u$. If Player 1 never enters  in the open set $\mathcal{L}=V\cdot \Gao$ then the final word written by Player 1 will be in $\mathcal{L}'$ and thus surely outside $L'$, and the final word written by Player 2 will be outside the empty set. So in that case Player 2 wins this play too.
If at some step of the play Player  1 enters  in the open set $\mathcal{L}=V\cdot \Gao$ then his final $\om$-word will be surely in $L'$. In that case Player 1, in charge of the set
$\emptyset + (\emptyset + L)$, can again write an extra letter and choose to be  in charge of the whole set and he can
now write an $\om$-word $v$ so that his final $\om$-word will be inside $\emptyset + (\emptyset + L)$. Thus Player 2 wins this play too.

 Finally we have proved that $L \leq_W  L'   \leq_W L $ thus it holds that $L' \equiv_W L$.  This ends the proof.
\end{proof}

{\bf End of Proof of Theorem \ref{thewad}. }

Let  $L \subseteq \Sio$ be a   Borel set accepted by a real time
$1$-counter B\"uchi automaton  $\mathcal{A}$.
If the Wadge degree of $L$ is finite, it is well known that it is Wadge
equivalent to a regular $\om$-language, hence also to an  $\om$-language in  the class {\bf r}-${\bf BCL}(4)_\om$.
If  $L$ is non self dual and its Wadge degree is greater than or equal to $\om$, then we know from Lemma \ref{nad} that
there is an $\om$-language  $L'$ accepted by a   a real time $4$-blind counter B\"uchi  automaton such that $L \equiv_W L'$.

 It remains to consider the case of self dual Borel sets.
The alphabet $\Si$ being finite, a self dual Borel set $L$ is always Wadge equivalent to a Borel set
in the form $\Si_1\cdot L_1 \cup \Si_2\cdot L_2$, where $(\Si_1, \Si_2)$ form a partition of $\Si$,
and $L_1, L_2\subseteq \Sio$ are non self dual Borel sets such that
$L_1 \equiv_W L_2^-$.
Moreover $L_1$ and $L_2$ can be taken in the form $L_{(u_1)}=u_1\cdot \Sio \cap L$ and
 $L_{(u_2)}=u_2\cdot \Sio \cap L$     for some $u_1, u_2 \in \Sis$, see
\cite{Duparc03}.
So if  $L \subseteq \Sio$ is a self dual Borel set accepted by a real time
$1$-counter B\"uchi automaton
then $L \equiv_W \Si_1\cdot L_1 \cup \Si_2\cdot L_2$, where $(\Si_1, \Si_2)$ form a partition of $\Si$, and
 $L_1, L_2\subseteq \Sio$ are non self dual Borel sets accepted by real time
$1$-counter B\"uchi automata.
We have already proved that there is  an     $\om$-language      $L'_1$  in  the class {\bf r}-${\bf BCL}(4)_\om$ such that $L'_1 \equiv_W L_1$ and
 an $\om$-language $L'_2$ in  the class {\bf r}-${\bf BCL}(4)_\om$  such that $L_2'^-\equiv_W L_2$.  Thus
$L  \equiv_W  \Si_1\cdot L_1 \cup \Si_2\cdot L_2  \equiv_W \Si_1\cdot L_1' \cup \Si_2\cdot L'_2$ and
$\Si_1\cdot L'_1 \cup \Si_2\cdot L'_2$ is an
$\om$-language in  the class {\bf r}-${\bf BCL}(4)_\om$.

The reverse direction is immediate: if $L \subseteq \Sio$ is a   Borel set accepted by a  $4$-blind counter B\"uchi  automaton $\mathcal{A}$, then it is also accepted by a B\"uchi Turing machine and thus by \cite[Theorem 25]{finkel2006borel} there exists a real time
$1$-counter B\"uchi automaton  $\mathcal{B}$ such that $\lang(\mathcal{A})  \equiv_W \lang(\mathcal{B})$.

This concludes the proof of Theorem \ref{thewad}.

Recall that, for each non-null countable ordinal $\alpha$,  the ${\bf \Si}^0_\alpha$-complete sets
(respectively, the  ${\bf \Pi}^0_\alpha$-complete  sets) form a single Wadge degree.  Thus  we can infer the following result from the above Theorem \ref{thewad} and from    the results of
 \cite{finkel2006borel,KMS89}.

\begin{corollary}
For each non-null recursive ordinal $\alpha < \om_1^{{\rm CK}} $ there
exist some ${\bf \Si}^0_\alpha$-complete and some
 ${\bf \Pi}^0_\alpha$-complete   $\om$-languages  in the class  {\bf r}-${\bf BCL}(4)_\om$.
And the supremum
of the set of Borel ranks of $\om$-languages  in the class  {\bf r}-${\bf BCL}(4)_\om$ is the ordinal $\gamma_2^1$, which  is precisely defined in \cite{KMS89}.
\end{corollary}

 We have only considered Borel sets in the above Theorem \ref{thewad}. However we know that there also exist  some non-Borel $\om$-languages accepted by
real time
$1$-counter B\"uchi automata, and even some ${\bf \Si}_1^1$-complete ones, \cite{Fin03a}.

   By  Lemma 4.7 of \cite{Fin13-JSL}
the conclusion of the above Lemma \ref{sum-wad} is also true if
$L$ is assumed to be an analytic but non-Borel set.

\begin{lemma}[\cite{Fin13-JSL}]\label{sum-wad2}
Let $L \subseteq \Sio$ be an analytic but non-Borel set. Then  $L \equiv_W \emptyset + L$.
\end{lemma}

\noindent Next  the proof of the above Lemma  \ref{nad} can be adapted to the case of an  analytic but non-Borel set, and we can state the following result.

\begin{theorem}\label{nad2}
Let  $L \subseteq \Sio$ be an analytic but non-Borel set accepted by a real time
$1$-counter B\"uchi  automaton   $\mathcal{A}$.
Then there is an $\om$-language $L'$ accepted by   a real time $4$-blind counter B\"uchi  automaton such that $L \equiv_W L'$.
\end{theorem}

\begin{proof}
It is very similar to the proof of the above Lemma  \ref{nad}, using Lemma \ref{sum-wad2} instead of the above Lemma \ref{sum-wad}.
\end{proof}

\begin{remark}\label{rem-nonborel}
Using  Lemma \ref{sum-wad2} instead of  the above Lemma \ref{sum-wad},  the proofs of \cite{finkel2006borel} can also be adapted  to the case of a non-Borel set to show that
     for every effective analytic but non-Borel set $L \subseteq \Sio$, where $\Si$ is a finite alphabet, there exists an $\om$-language $L'$ in
 {\bf r}-${\bf CL}(1)_\om$ such that $L' \equiv_W L$, and thus also, by Theorem \ref{nad2},  an $\om$-language $L''$ accepted by   a real time $4$-blind counter B\"uchi  automaton
such that $L \equiv_W L''$.
\end{remark}

 This   implies in particular the existence of a ${\bf \Si}_1^1$-complete, hence non Borel,  $\om$-language accepted by a  real-time $4$-blind-counter B\"uchi automaton.

\begin{corollary}
  There exists a ${\bf \Si}_1^1$-complete $\om$-language  accepted by a 4-blind-counter automaton.
 \end{corollary}

Notice that  if we assume the axiom of ${\bf \Si}_1^1$-determinacy, then   any set which is analytic but not Borel is
${\bf \Si}_1^1$-complete, see \cite{Kechris94}, and thus  there is only one more Wadge degree  (beyond Borel sets) containing  ${\bf \Si}_1^1$-complete sets.
On the other hand,  if the  axiom of  (effective) ${\Si}_1^1$-determinacy does not hold, then there exist some effective analytic sets which are neither Borel nor
${\bf \Si}_1^1$-complete.
Recall that  {\bf ZFC}  is  the commonly accepted axiomatic
framework for Set Theory in which all usual mathematics can be developed.

\begin{corollary}\label{non-complete}
It is consistent with  {\bf ZFC}  that there exist some
$\om$-languages of Petri nets  in the class {\bf r}-${\bf BCL}(4)_\om$     which are neither Borel nor  ${\bf \Si}_1^1$-complete.
\end{corollary}

\begin{proof}
Recall that   {\bf ZFC}  is  the commonly accepted axiomatic
framework for Set Theory in which all usual mathematics can be developed.  The determinacy of Gale-Stewart games $G(A)$,  where
$A$ is an (effective) analytic set, denoted {\bf Det}($\Si_1^1$),  is not provable in  {\bf ZFC}; Martin and Harrington have proved that it is  a large cardinal
assumption equivalent to the existence of a particular real, called the real $0^\sharp$, see \cite[page  637]{Jech}.
It is also known that the determinacy of (effective) analytic
Gale-Stewart games is equivalent to the determinacy of (effective) analytic Wadge games, denoted  {\bf W-Det}($\Si_1^1$), see \cite{Louveau-Saint-Raymond}.

It is known that, if   {\bf ZFC}  is consistent, then there is a model of   {\bf ZFC}  in which the determinacy of (effective) analytic
Gale-Stewart games, and thus also the determinacy of (effective) analytic Wadge games, do not hold. It follows from \cite[Theorem 4.3]{harrington} that in such a model of   {\bf ZFC}  there exists an effective analytic set which is neither Borel nor
${\bf \Si}_1^1$-complete. The result now follows from Theorem \ref{nad2}  and Remark \ref{rem-nonborel}.
\end{proof}

We can now get  an amazing result on $\om$-languages of Petri nets from the following  previous one which was obtained in \cite{Fin-ICST}.  Recall that in a model {\bf  V} of  {\bf ZFC}, we denote by  {\bf  L} the class of constructible sets of {\bf  V} which induces a submodel of {\bf ZFC}. The axiom {\bf V=L} expresses that ``every set is constructible" and it is consistent with {\bf ZFC}.  The ordinal $\om_1^{\bf L}$ denotes the first uncountable ordinal in   the model {\bf  L}. It is consistent that  $\om_1^{\bf L} = \om_1$ since this is true in the model {\bf  L}. But it is also consistent with {\bf ZFC} that $\om_1^{\bf L} < \om_1$. We refer the interested reader to \cite{Fin-ICST}, and to classical textbooks on set theory like \cite{Jech} for more information about these notions.

The following result was obtained in \cite{Fin-ICST}:

\begin{theorem}\label{top-zfc}
One can construct a real-time $1$-counter B\"uchi automaton  $\mathcal{A}$ such that the topological complexity of the
$\om$-language $\lang(\mathcal{A})$ is not determined by the axiomatic system {\bf ZFC}. Indeed it holds that:
\begin{enumerate}
\itemsep=0.88pt
\item ({\bf ZFC + V=L}). ~~~~~~ The $\om$-language $\lang(\mathcal{A})$ is  analytic but not Borel.
\item ({\bf ZFC} + $\om_1^{\bf L} < \om_1$).  ~~~~The $\om$-language $\lang(\mathcal{A})$ is a  ${\bf \Pi}^0_2$-set.
\end{enumerate}
\end{theorem}

We can now easily get the following result:

\begin{theorem}\label{top-zfc2}
One can construct a $4$-blind counter B\"uchi automaton  $\mathcal{B}$ such that the topological complexity of the
$\om$-language $\lang(\mathcal{B})$ is not determined by the axiomatic system {\bf ZFC}. Indeed it holds that:
\begin{enumerate}
\itemsep=0.88pt
\item ({\bf ZFC + V=L}). ~~~~~~ The $\om$-language $\lang(\mathcal{B})$ is  analytic but not Borel.
\item ({\bf ZFC} + $\om_1^{\bf L} < \om_1$).  ~~~~The $\om$-language $\lang(\mathcal{B})$ is a  ${\bf \Pi}^0_2$-set.
\end{enumerate}
\end{theorem}

\begin{proof}
It follows directly from Theorem \ref{top-zfc} and from  the construction of the real time $4$-blind counter B\"uchi  automaton $\mathcal{P}_{\mathcal{A}}$ such that
$\lang(\mathcal{P}_{\mathcal{A}}) = h( \lang(\mathcal{A}) ) \cup \mathcal{L}$,  which was described above. It then suffices to take $\mathcal{B}=\mathcal{P}_{\mathcal{A}}$, where  $\mathcal{A}$ is the  real-time $1$-counter B\"uchi automaton constructed in the proof of  Theorem~\ref{top-zfc}.
\end{proof}

\section{High undecidability of topological and arithmetical properties}
\label{sec:high-undec-top}

    We prove that it is highly undecidable to determine the topological complexity of a Petri net $\om$-language.  As usual, since there is a finite description of a  real time
$1$-counter B\"uchi automaton or of a 4-blind-counter  B\"uchi automaton, we can define a G\"odel  numbering of all $1$-counter B\"uchi automata or of all
 4-blind-counter  B\"uchi automata and then speak about the $1$-counter B\"uchi automaton (or 4-blind-counter  B\"uchi automaton) of index $z$.
  Recall first the following result,  proved in~\cite{Fin-HI}, where
  we denote  $\mathcal{A}_z$ the real time
$1$-counter B\"uchi automaton of index $z$ reading words over a fixed finite alphabet $\Si$ having at least two letters.
We refer the reader to a textbook like \cite{Odifreddi1}
for more background about the
analytical hierarchy of subsets of the set $\mathbb{N}$ of natural numbers.

\begin{theorem}\label{borel-hard}
Let $\alpha$ be a countable ordinal. Then
\begin{enumerate}
\itemsep=0.88pt
\item $ \{  z \in \mathbb{N}  \mid  \lang(\mathcal{A}_z) \mbox{ is in the Borel class } {\bf \Si}^0_\alpha \}$ is  $\Pi_2^1$-hard.
\item  $ \{  z \in \mathbb{N}  \mid  \lang(\mathcal{A}_z) \mbox{ is in the Borel class } {\bf \Pi}^0_\alpha \}$ is  $\Pi_2^1$-hard.
\item  $ \{  z \in \mathbb{N}  \mid  \lang(\mathcal{A}_z) \mbox{ is a  Borel set } \}$ is  $\Pi_2^1$-hard.
\end{enumerate}
\end{theorem}

 Using the previous constructions we can now easily show the following result, where $\mathcal{P}_z$ is the real time 4-blind-counter  B\"uchi automaton of index $z$.

\begin{theorem}\label{borel-hard-pn}
Let $\alpha\geq 2$ be a countable ordinal. Then
\begin{enumerate}
\itemsep=0.88pt
\item $ \{  z \in \mathbb{N}  \mid  \lang(\mathcal{P}_z) \mbox{ is in the Borel class } {\bf \Si}^0_\alpha \}$ is  $\Pi_2^1$-hard.
\item  $ \{  z \in \mathbb{N}  \mid  \lang(\mathcal{P}_z) \mbox{ is in the Borel class } {\bf \Pi}^0_\alpha \}$ is  $\Pi_2^1$-hard.
\item  $ \{  z \in \mathbb{N}  \mid  \lang(\mathcal{P}_z) \mbox{ is a  Borel set } \}$ is  $\Pi_2^1$-hard.
\end{enumerate}
\end{theorem}

\begin{proof}
It follows  from the fact that one can easily get an injective  recursive function $g: \mathbb{N} \ra \mathbb{N}$ such that
$\mathcal{P}_{\mathcal{A}_z} =  h( \lang(\mathcal{A}_z) ) \cup \mathcal{L} = \lang(\mathcal{P}_{g(z)}) $ and from the following equivalences which hold for each countable ordinal
$\alpha \geq 2$:
\begin{enumerate}
\itemsep=0.88pt
\item  $\lang(\mathcal{A}_z)$  is in the Borel class  ${\bf \Si}^0_\alpha$  (resp.,  ${\bf \Pi}^0_\alpha$)   $\Longleftrightarrow$ $ \lang(\mathcal{P}_{g(z)})$  is in the Borel class  ${\bf \Si}^0_\alpha$
(resp., ${\bf \Pi}^0_\alpha$).
\item  $\lang(\mathcal{A}_z)$  is a  Borel set  $\Longleftrightarrow$ $ \lang(\mathcal{P}_{g(z)})$  is a  Borel  set.
\end{enumerate}

\vspace*{-8mm}
\end{proof}

Recall that the arithmetical properties of $\om$-languages of  real time $1$-counter B\"uchi automata were also proved to be highly undecidable  in \cite{Fin-HI}.

\begin{theorem}\label{arithmetic-dec}
\noindent Let $n \geq 1$ be an integer. Then
\begin{enumerate}
\itemsep=0.88pt
\item $ \{  z \in \mathbb{N}  \mid  \lang(\mathcal{A}_z) \mbox{ is in the arithmetical class }  \Si_n \}$ is  $\Pi_2^1$-complete.
\item $ \{  z \in \mathbb{N}  \mid  \lang(\mathcal{A}_z) \mbox{ is in the arithmetical class }  \Pi_n \}$ is  $\Pi_2^1$-complete.
\item $ \{  z \in \mathbb{N}  \mid  \lang(\mathcal{A}_z) \mbox{ is a }  \Delta^1_1 \mbox{-set } \}$ is  $\Pi_2^1$-complete.
\end{enumerate}
\end{theorem}

We can now prove  similar results for $\om$-languages of  real time 4-blind-counter  B\"uchi automata.

\begin{theorem}
\noindent  Let $n \geq 2$ be an integer. Then
\begin{enumerate}
\itemsep=0.88pt
\item $ \{  z \in \mathbb{N}  \mid  \lang(\mathcal{P}_z) \mbox{ is in the arithmetical class }  \Si_n \}$ is  $\Pi_2^1$-complete.
\item $ \{  z \in \mathbb{N}  \mid  \lang(\mathcal{P}_z) \mbox{ is in the arithmetical class }  \Pi_n \}$ is  $\Pi_2^1$-complete.
\item $ \{  z \in \mathbb{N}  \mid  \lang(\mathcal{P}_z) \mbox{ is a }  \Delta^1_1 \mbox{-set } \}$ is  $\Pi_2^1$-complete.
\end{enumerate}
\end{theorem}

\eject

\begin{proof}
Firstly, the three sets of integers considered in this theorem can be seen to be in the class $\Pi_2^1$.  This can be proved in a very similar way as in the proof of \cite[Theorem 3.26]{Fin-HI}.
Secondly, the completeness part of the results follows from Theorem \ref{arithmetic-dec} and from the fact that one can easily get an injective  recursive function $g\colon \mathbb{N} \ra \mathbb{N}$ such that
$\mathcal{P}_{\mathcal{A}_z} =  h( \lang(\mathcal{A}_z) ) \cup \mathcal{L} = \lang(\mathcal{P}_{g(z)}) $ and from the following equivalences which hold for each integer
$n \geq 2$, due to the fact that $\mathcal{L}$ is an arithmetical $\Sigma^0_1$-set and $h$ is a recursive function from $\Sio$ onto the effective closed set $h(\Sio)$.
\begin{enumerate}
\itemsep=0.88pt
\item  $\lang(\mathcal{A}_z)$  is in the  arithmetical class $ \Si_n$    (resp.,  $\Pi_n$)   $\Longleftrightarrow$ $ \lang(\mathcal{P}_{g(z)})$  is in the arithmetical class  $ \Si_n$    (resp.,  $\Pi_n$)
\item  $\lang(\mathcal{A}_z)$  is a  $\Delta^1_1 $-set  $\Longleftrightarrow$ $ \lang(\mathcal{P}_{g(z)})$  is a   $\Delta^1_1 $-set.
\end{enumerate}

\vspace*{-6mm}
\end{proof}

\section{High undecidability of the equivalence and the inclusion problems}
\label{sec:high-undec-equiv}

We now add a result obtained from our  previous constructions and which is  important for  verification purposes.

\begin{theorem}\label{}
The equivalence and the inclusion problems for $\om$-languages of Petri nets, or even for   $\om$-languages in the class   {\bf r}-${\bf BCL}(4)_\om$,    are $\Pi_2^1$-complete.
\begin{enumerate}
\item $\{ (z,z')\in \mathbb{N} \mid \lang(\mathcal{P}_z) = \lang(\mathcal{P}_{z'}) \}$ is $\Pi_2^1$-complete
\item $\{ (z,z')\in \mathbb{N} \mid \lang(\mathcal{P}_z) \subseteq \lang(\mathcal{P}_{z'}) \}$ is $\Pi_2^1$-complete
\end{enumerate}

\end{theorem}

\begin{proof}
Firstly, it is easy to see that each of these decision problems is in the class $\Pi_2^1$, since the equivalence and the inclusion problems for $\om$-languages of Turing machines are already in the class $\Pi_2^1$, see \cite{cc,Fin-HI}.
The completeness part follows from the fact that the  equivalence and the inclusion problems for $\om$-languages accepted by
real time $1$-counter B\"uchi automata are $\Pi_2^1$-complete  \cite{Fin-HI},  and  from the fact that there exists an injective  recursive function $g: \mathbb{N} \ra \mathbb{N}$ such that
$\mathcal{P}_{\mathcal{A}_z}=\mathcal{P}_{g(z)}$, and then from the following equivalences:
\begin{enumerate}
\item  $\lang(\mathcal{A}_z) = \lang(\mathcal{A}_{z'})   \Longleftrightarrow  \lang(\mathcal{P}_{g(z)}) = \lang(\mathcal{P}_{g(z')})$
\item  $\lang(\mathcal{A}_z) \subseteq \lang(\mathcal{A}_{z'})
\Longleftrightarrow \lang(\mathcal{P}_{g(z)}) \subseteq \lang(\mathcal{P}_{g(z')})$
\end{enumerate}

\vspace*{-6mm}
\end{proof}

\section{Determinacy of Wadge games}
\label{sec:wadge-determ}

We proved in \cite{Fin13-JSL}  that the determinacy of Wadge games between two players in charge of
$\om$-languages accepted by  real time $1$-counter B\"uchi automata,  denoted    {\bf W-Det}({\bf r}-${\bf CL}(1)_\om$),   is  equivalent to the  (effective) analytic  Wadge determinacy.

We can now state the following result, proved within the axiomatic system {\bf ZFC}.

\begin{theorem}
The determinacy of Wadge games between two players in charge of $\om$-languages in the class  {\bf r}-${\bf BCL}(4)_\om$ is equivalent to the effective analytic (Wadge) determinacy, and thus  is not provable in  the axiomatic system {\bf ZFC}.
\end{theorem}

\begin{proof}
It was proved in \cite{Fin13-JSL} that  the following equivalence holds: {\bf W-Det}({\bf r}-${\bf CL}(1)_\om$) $\Longleftrightarrow$  {\bf W-Det}($\Si_1^1$).
The implication {\bf W-Det}($\Si_1^1$)$\Longrightarrow$ {\bf W-Det}({\bf r}-${\bf BCL}(4)_\om$) is obvious since the class ${\bf BCL}(4)_\om$  is included into the class $\Si_1^1$.
To prove the reverse implication, we assume that {\bf W-Det}({\bf r}-${\bf BCL}(4)_\om$) holds and we show that every Wadge game $W(\lang(\mathcal{A}),\lang(\mathcal{B}))$  between two players in charge of
$\om$-languages of the class {\bf r}-${\bf CL}(1)_\om$ is determined (we assume without loss of generality that the two real time 1-counter  B\"uchi automata $\mathcal{A}$ and $\mathcal{B}$ read words over the same alphabet $\Si$).

It is sufficient to consider the cases where at least one of two $\om$-languages $\lang(\mathcal{A})$ and $\lang(\mathcal{B})$ is non-Borel, since the Borel Wadge determinacy is provable in
{\bf ZFC}.  On the other hand, we have seen how we can construct some real time 4-blind-counter  B\"uchi automata $\mathcal{P}_{\mathcal{A}}$ and $\mathcal{P}_{\mathcal{B}}$
such that $\lang(\mathcal{P}_{\mathcal{A}})=h(\lang(\mathcal{A})) \cup \mathcal{L}$  and $\lang(\mathcal{P}_{\mathcal{B}})=h(\lang(\mathcal{B})) \cup \mathcal{L}$.

We can firstly  consider the case where  $\lang(\mathcal{A})$ is Borel of Wadge degree smaller than $\om$, and $\lang(\mathcal{B})$  is non-Borel. In that case $\lang(\mathcal{A})$ is in particular a
${\bf \Pi}^0_2$-set.  Recall now that we can infer from Hurewicz's Theorem, see \cite[page 160]{Kechris94},
that an analytic subset of $\Si^\om$  is either ${\bf \Pi}^0_2$-hard   or  a ${\bf \Si}^0_2$-set. Thus $\lang(\mathcal{B})$  is ${\bf \Pi}^0_2$-hard and Player 2 has a winning strategy
in the game  $W(\lang(\mathcal{A}),\lang(\mathcal{B}))$.

Secondly we consider the case where  $\lang(\mathcal{A})$ and $\lang(\mathcal{B})$ are either non-Borel or Borel of Wadge degree greater  than $\om$.
By hypothesis we know that the Wadge game $W(\lang(\mathcal{P}_{\mathcal{A}}), \lang(\mathcal{P}_{\mathcal{B}}))$ is determined, and that one of the players has a winning strategy.
Using the above constructions and reasonings we used in the proofs of Lemmas \ref{lem8} and \ref{nad},
we can easily show that  the same player has a winning strategy in the Wadge game $W(\lang(\mathcal{A}),\lang(\mathcal{B}))$.

 We now consider the two following cases:\medskip
\nl {\bf First case.} Player 2 has a w.s. in the game $W(\lang(\mathcal{P}_{\mathcal{A}}), \lang(\mathcal{P}_{\mathcal{B}}))$. If $\lang(\mathcal{B})$ is Borel then $\lang(\mathcal{P}_{\mathcal{B}})$ is easily seen to be Borel
 and then $\lang(\mathcal{P}_{\mathcal{A}})$ is also Borel because $\lang(\mathcal{P}_{\mathcal{A}}) \leq_W    \lang(\mathcal{P}_{\mathcal{B}})$.  Thus  $\lang(\mathcal{A})$ is also Borel and  the game
$W(\lang(\mathcal{A}), \lang(\mathcal{B}))$ is determined. Assume now that $\lang(\mathcal{B})$ is not Borel. Consider the Wadge game $W(\lang(\mathcal{A}), \emptyset + (\emptyset + \lang(\mathcal{B})))$. We claim that Player 2 has a w.s. in that
game which is easily deduced from a w.s. of Player 2 in the Wadge game  $W(\lang(\mathcal{P}_{\mathcal{A}}), \lang(\mathcal{P}_{\mathcal{B}}))$ $=W(h(\lang(\mathcal{A})) \cup \mathcal{L}, h(\lang(\mathcal{B})) \cup \mathcal{L})$. Consider a play in this latter game where  Player
1 remains in the closed set $h(\Si^\om)$:  she writes a beginning of a word in the form
$$A0x(1)B0^{2}x(2)A \cdots
B0^{2n}x(2n)A \cdots  $$
\noindent Then player 2 writes a beginning of a word in the form
$$A0x'(1)B0^{2}x'(2)A \cdots
B0^{2p}x'(2p)A \cdots $$
\noindent where $p\leq n$.
Then the strategy of Player 2 in $W(\lang(\mathcal{A}), \emptyset + (\emptyset + \lang(\mathcal{B})))$ consists to write $x'(1)x'(2) \cdots \allowbreak
x'(2p)$ when Player 1 writes  $x(1)x(2) \cdots x(2n)$. (Notice that Player 2 is allowed to skip, provided he really writes an $\om$-word in $\om$ steps).
If the strategy for Player 2 in $W(\lang(\mathcal{P}_{\mathcal{A}}), \lang(\mathcal{P}_{\mathcal{B}}))$ was at some step to go out of
the  closed set $h(\Si^\om)$ then this means that   the word he has now
written is not a prefix of any $\om$-word of  $h(\Sio)$,  and  Player 2 ``enters'' in the open set $h(\Sio)^- = \mathcal{L}\cup \mathcal{L}'$ and will stay in this set. Two subcases may now appear.

\medskip
{\bf Subcase A.} When  Player 2  in the game $W(\lang(\mathcal{P}_{\mathcal{A}}), \lang(\mathcal{P}_{\mathcal{B}}))$  ``enters'' in the open set $h(\Sio)^- = \mathcal{L}\cup \mathcal{L}'$, he actually enters in the open set $ \mathcal{L}$. Then the final word written by Player 2 will surely be inside $\lang(\mathcal{P}_{\mathcal{B}})$. Player 2 in the Wadge game $W(\lang(\mathcal{A}), \emptyset + (\emptyset + \lang(\mathcal{B})))$ can now  write a letter of $\Delta -\Si$ in such a way that he is now like a player in charge of the whole set  and he can
now write an $\om$-word $u$ so that his final $\om$-word will be inside $\emptyset +  (\emptyset + \lang(\mathcal{B}))$. Thus Player 2 wins this play too.

\medskip
{\bf  Subcase B.}   When  Player 2  in the game $W(\lang(\mathcal{P}_{\mathcal{A}}), \lang(\mathcal{P}_{\mathcal{B}}))$  ``enters'' in the open set $h(\Sio)^- = \mathcal{L}\cup \mathcal{L}'$, he does not enter in the open set $ \mathcal{L}$.
Then  Player 2, in the Wadge game $W(\lang(\mathcal{A}), \emptyset + (\emptyset + \lang(\mathcal{B})))$,  being first like a player in charge of the set $(\emptyset +  \lang(\mathcal{B}))$,  can write a  letter of $\Delta -\Si$ in such a way that he is now like a player in charge of the empty set  and he can
now continue, writing an $\om$-word $u$. If Player 2  in the game $W(\lang(\mathcal{P}_{\mathcal{A}}), \lang(\mathcal{P}_{\mathcal{B}}))$  never enters  in the open set $\mathcal{L}$ then the final word written by Player 2 will be in $\mathcal{L}'$ and thus surely outside $\lang(\mathcal{P}_{\mathcal{B}})$, and the final word written by Player 2 will be outside the empty set. So in that case Player 2 wins this play too in the Wadge game $W(\lang(\mathcal{A}), \emptyset + (\emptyset + \lang(\mathcal{B})))$.
If at some step of the play, in the game $W(\lang(\mathcal{P}_{\mathcal{A}}), \lang(\mathcal{P}_{\mathcal{B}}))$,  Player  2 enters  in the open set $\mathcal{L}$ then his final $\om$-word will be surely in $\lang(\mathcal{P}_{\mathcal{B}})$. In that case Player 2, in charge of the set
$\emptyset + (\emptyset + \lang(\mathcal{B}))$   in the Wadge game $W(\lang(\mathcal{A}), \emptyset + (\emptyset + \lang(\mathcal{B})))$, can again write an extra letter and choose to be  in charge of the whole set and he can
now write an $\om$-word $v$ so that his final $\om$-word will be inside $\emptyset + (\emptyset +  \lang(\mathcal{B}))$. Thus Player 2 wins this play too.

  So we have proved that Player 2 has a w.s. in  the Wadge game $W(\lang(\mathcal{A}), \emptyset + (\emptyset + \lang(\mathcal{B})))$  or equivalently that
$\lang(\mathcal{A}) \leq_W   \emptyset + (\emptyset + \lang(\mathcal{B}))$. But by  Lemma \ref{sum-wad2} we know that  $\lang(\mathcal{B})  \equiv_W \emptyset + (\emptyset + \lang(\mathcal{B}))$  and thus
$\lang(\mathcal{A}) \leq_W \lang(\mathcal{B})$ which means that Player 2 has a  w.s. in  the Wadge game $W(\lang(\mathcal{A}), \lang(\mathcal{B}))$.

\medskip
{\bf Second case.} Player 1 has a w.s. in the game $W(\lang(\mathcal{P}_{\mathcal{A}}), \lang(\mathcal{P}_{\mathcal{B}}))$. Notice that this implies that
$ \lang(\mathcal{P}_{\mathcal{B}}) \leq_W \lang(\mathcal{P}_{\mathcal{A}})^-$. Thus if $\lang(\mathcal{A})$ is Borel then $\lang(\mathcal{P}_{\mathcal{A}})$ is Borel, $\lang(\mathcal{P}_{\mathcal{A}})^-$ is also Borel,
and  $ \lang(\mathcal{P}_{\mathcal{B}})$ is Borel as the inverse image of a Borel set by a continuous function, and $ \lang(\mathcal{B})$ is also Borel, so the Wadge game
$W( \lang(\mathcal{A}),  \lang(\mathcal{B}))$ is determined. We now  assume that $\lang(\mathcal{A})$ is not Borel and we consider the Wadge game $W(\lang(\mathcal{A}),  \lang(\mathcal{B}))$.
Player 1 has a w.s. in this game which is easily constructed from a w.s. of the same player in the game  $W(\lang(\mathcal{P}_{\mathcal{A}}), \lang(\mathcal{P}_{\mathcal{B}}))$ as follows.
For this consider a play in this latter game where Player 2 does not go out of the closed set $h(\Sio)$.

\medskip
He writes a beginning of a word in the form
$$A0x(1)B0^{2}x(2)A \cdots
B0^{n}x(n)A \cdots  $$
\noindent Then Player 1 writes a beginning of a word in the form
$$A0x'(1)B0^{2}x'(2)A \cdots
B0^{p}x'(p)A \cdots $$
where $n \leq p$ (notice that here without loss of generality  the notation implies that $n$ and $p$ are even, since the segments $B0^{n}x(n)A$ and $B0^{p}x'(p)A$ begin
 with a letter $B$      but this is not essential in the proof).
Then the strategy for Player 1 in $W(\lang(\mathcal{A}),  \lang(\mathcal{B}))$ consists to write $x'(1)x'(2) \cdots
x'(p)$ when Player 2 writes  $x(1)x(2) \cdots x(n)$.   After $\om$ steps, the $\om$-word written by Player 1 is in $\lang(\mathcal{A})$ iff the $\om$-word written by Player 2 is not in the set
 $\lang(\mathcal{B})$, and thus Player 1 wins the play.

If the strategy for Player 1 in $W(\lang(\mathcal{P}_{\mathcal{A}}), \lang(\mathcal{P}_{\mathcal{B}}))$ was at some step to go out of
the  closed set $h(\Si^\om)$ then this means that  she ``enters'' in the open set $h(\Sio)^- = \mathcal{L}\cup \mathcal{L}'$ and will stay in this set. Two subcases may now appear.

\medskip
{\bf Subcase A.} When  Player 1  in the game $W(\lang(\mathcal{P}_{\mathcal{A}}), \lang(\mathcal{P}_{\mathcal{B}}))$  ``enters'' in the open set $h(\Sio)^- = \mathcal{L}\cup \mathcal{L}'$, she actually enters in the open set $ \mathcal{L}$. Then the final word written by Player 1 will surely be inside $\lang(\mathcal{P}_{\mathcal{A}})$. But she wins the play since she follows a winning strategy and this leads to a contradiction. Indeed if Player 2 decided to also enter in  in the open set $ \mathcal{L}$ then Player 2 would win the play. Thus this case is actually not possible.

\medskip
{\bf  Subcase B.}   When  Player 1  in the game $W(\lang(\mathcal{P}_{\mathcal{A}}), \lang(\mathcal{P}_{\mathcal{B}}))$  ``enters'' in the open set $h(\Sio)^- = \mathcal{L}\cup \mathcal{L}'$, she does not enter in the open set $ \mathcal{L}$. But Player 2 would be able to do the same and enter in $h(\Sio)^- = \mathcal{L}\cup \mathcal{L}'$ but not (for the moment) in  the open set $ \mathcal{L}$. And if at some step of the play, Player 1 would enter in the open set $ \mathcal{L}$ then Player 2 could do the same, and thus Player 2 would win the play.
Again this is not possible since Player 1 wins the play since she follows a winning strategy.

 Finally both subcases A and B cannot occur and this shows that Player 1 has a w.s. in the Wadge game $W(\lang(\mathcal{A}), \lang(\mathcal{B}))$.
\end{proof}

\section{Non-Borel $\om$-languages of one counter Petri nets}
\label{sec:non-borel}

In this section we prove the following result.

\begin{theorem}
There exists a $1$-blind counter automaton $\Aa_1$ such that the  $\om$-language $\lang(\mathcal{A}_1)$ of that automaton is  $\asigma{1}$\=/complete, hence non-Borel.
\end{theorem}

The crucial obstacle of that construction is the fact that the simulation order ${\preceq}$ for $1$-blind counter automata has a finite width, namely the number of states of the machine. This means that the non-deterministic choices of such a machine are inherently ordered. This explains why we use a $\asigma{1}$-hard $\omega$-language that itself is a set of orders.

\subsection{Topology of orders}
\label{ssec:top-orders}

Consider a~set $X$ and a~relation ${\orrA}\subseteq X\times X$ on $X$. We say that ${\orrA}$ is a~\emph{linear order} if it is reflexive, transitive, and anti\=/symmetric. We interpret a~pair $(x,x')\in \orrA$ as representing the fact that $x$ is $\orrA$-smaller-or-equal than $x'$. A~linear order ${\orrA}$ is \emph{ill\=/founded} if there exists an infinite sequence $x_0,x_1,\ldots$ of pairwise distinct elements of $X$ such that for all $n$ we have $(x_{n+1},x_n)\in\orrA$, i.e.~the sequence is strictly decreasing. Such a~sequence indicates an~infinite ${\orrA}$\=/descending chain. An~order that is not ill\=/founded is called \emph{well\=/founded}.

The binary tree is the set of all sequences of \emph{directions} $\Tt\eqdef\{\dL,\dR\}^\ast$ where the \emph{directions} $\dL$ and $\dR$ are two fixed distinct symbols. For technical reasons we sometimes consider a~third direction $\dM$ (it does not occur in the binary tree).

A set $X\subseteq \Tt$ can be naturally identified with its characteristic function $X\in \{0,1\}^{\big(\{\dL,\dR\}^\ast\big)}$. Thus, the family of all subsets of the binary tree, with the natural product topology, is homeomorphic with the Cantor set $\{0,1\}^\w$.

\medskip
The elements $\verA,\verB\in\Tt$ are called \emph{nodes}. Nodes are naturally ordered by the following three orders:
\begin{itemize}
\itemsep=0.9pt
\item the prefix order ${\sqsubseteq}$, as defined on page~\pageref{def:prefix},
\item the lexicographic order: $\verA\lexeq \verB$ if $\verA$ is lexicographically smaller than $\verB$ (we assume that $\dL\lex \dM \lex \dR$),
\item the infix order: $\verA\inxeq \verB$ if $\verA \dM^\w$ (i.e.~the $\w$\=/word obtained by appending infinitely many symbols $\dM$ after $\verA$) is lexicographically less or equal than $\verB \dM^\w$. This order corresponds to the left\=/to\=/right order on the nodes of the tree, when it is drawn in a standard way.
\end{itemize}

Notice that, for every fixed $n$, when restricted to $\{\dL,\dR\}^n$, the lexicographic and infix orders coincide. However, $\dL \inx \epsilon\inx \dR$ but $\epsilon$ is the minimal element of ${\lexeq}$. Both the lexicographic and infix orders are linear.

Since the infix order is countable, dense, and has no minimal nor maximal elements, we obtain the following fact.

\begin{fact}
$(\Tt,{\inxeq})$ is isomorphic to the order of rational numbers $(\mathbb{Q},{\leq})$.
\end{fact}

In the following part of the paper we will use the following set.
\begin{equation}
\IFinx \eqdef \{X\subseteq\Tt\mid \text{$X$ contains an~infinite ${\inxeq}$\=/descending chain}\}.
\end{equation}

\begin{lemma}
The set $\IFinx$ is $\asigma{1}$\=/complete.
\end{lemma}

\begin{proof}
$\IFinx$ belongs to $\asigma{1}$ just by the form of the definition.

We prove $\asigma{1}$\=/hardness by a~reduction from the set of ill\=/founded linear orders on $\w$ (seen as elements of $\{0,1\}^{\w\times\w}$). Let us prove this fact more formally. Consider an~element $\orrA\in \{0,1\}^{\w\times\w}$ that is a~linear order on $\w$. The latter set is $\asigma{1}$\=/complete by a~theorem by Lusin and Sierpi{\'n}ski~\cite[Theorem~27.12]{Kechris94}. We will inductively define $X_{\orrA}\subseteq\Tt$ in such a~way to ensure that $\orrA\mapsto X_{\orrA}$ is a~continuous mapping and $\orrA$ is ill\=/founded if and only if $X_{\orrA}\in\IFinx$.

Let us proceed inductively, defining a~sequence of nodes $(x_n)_{n\in\w}\subseteq\Tt$. Our invariant says that $|x_k|=k$ and the map $k\mapsto x_k$ is an~isomorphism of the orders $\big(\{0,\ldots,n\}, \orrA\big)$ and $\big(\{x_0,\ldots,x_n\},\allowbreak{\inxeq}\big)$. We start with $x_0=\epsilon$ (i.e.~the root of $\Tt$). Assume that $x_0, \ldots, x_n$ are defined and satisfy the invariants. By the definition of ${\inxeq}$, there exists a~node $x\in\{\dL,\dR\}^{n+1}$ such that for $k=0,1,\ldots,n$ we have $x\inxeq x_k$ if and only if $(n{+}1,k)\in o$. Let $x_{n+1}$ be such a~node.

The above induction defines an~infinite sequence of nodes $x_0,x_1,\ldots$ Let $X_{\orrA}\eqdef\{x_n\mid n\in\w\}\subseteq\Tt$. By the definition of $X_{\orrA}$, the mapping $\orrA\mapsto X_{\orrA}$ is continuous --- the fact whether a~node $x\in\Tt$ belongs to $X_{\orrA}$ depends only on $\orrA\cap\{0,1,\ldots,|x|\}^2$. Using our invariant, we know that the map $k\mapsto x_k$ is an~isomorphism of the orders $\big(\w, \orrA\big)$ and $\big(X_{\orrA}, {\inxeq}\big)$. Thus, $\orrA$ is ill\=/founded if and only if $X_{\orrA}\in \IFinx$.
\end{proof}

To construct our continuous reduction in the one\=/counter case, we need the following simple lemma that provides an~alternative characterisation of the set~$\IFinx$. Let us introduce the following definition.

\begin{definition}
A sequence $\verA_0,\verA_1,\ldots\in \Tt$ is called a~\emph{correct chain} if $\verA_0=\epsilon$ and for every $n=0,1,\ldots$:
\begin{enumerate}
\itemsep=0.9pt
\item $|\verA_{n+1}|=|\verA_n|+1$,
\label{it:cor-ch-len}
\item $\verA_{n+1}\inxeq \verA_n \dR$ (or equivalently $\verA_{n+1}\lexeq \verA_n \dR$).
\label{it:cor-ch-ord}
\end{enumerate}

A correct chain is \emph{witnessing} for a~set $X\subseteq\Tt$ if for infinitely many $n$ we have $\verA_n \in X$ and $\verA_{n+1}\inxeq \verA_n \dL$.
\end{definition}

Intuitively, the definition forces the sequence to be \emph{not so\=/much increasing} in the infix order ${\inxeq}$: the successive element $\verA_{n+1}$ needs to be \emph{to the left} in the tree from $\verA_n\dR$. Such a~sequence is \emph{witnessing} for a set $X$ if infinitely many times it belongs to $X$ and at these moments it actually drops in ${\inxeq}$.

\begin{lemma}
A set $X\subseteq \Tt$ belongs to $\IFinx$ if and only if there exists a~correct chain witnessing for~$X$.
\end{lemma}

\begin{proof}
First take a~correct chain witnessing for $X$. Let $\verB_0,\verB_1,\ldots$ be the subsequence that shows that $(\verA_n)_{n\in\w}$ is witnessing for $X$. In that case, by the definition, for all $n$ we have $\verB_n\in X$ and $\verB_{n+1}\inx \verB_n$ (because $\verB_{n+1}\dM^\w\lexeq \verB_n\dL\dR^\w\lex \verB_n\dM^\w$). Thus, $X$ has an~infinite ${\inxeq}$\=/descending chain and belongs to~$\IFinx$.

\medskip
Now assume that $X\in\IFinx$ and $\verB_0\inxg\verB_1\inxg\verB_2\inxg\ldots$ is a~sequence witnessing that. Without loss of generality we can assume that $|\verB_{n+1}|>|\verB_n|$ because for each fixed depth $k$ there are only finitely many nodes of $\Tt$ in $\{\dL,\dR\}^{\leq k}$. We can now add intermediate nodes in\=/between the sequence $(\verB_n)_{n\in\w}$ to construct a~correct chain witnessing for $X$; the following pseudo\=/code realises this goal:\smallskip

\begin{minipage}{0.8\linewidth}
\begin{lstlisting}
$n$ := $0$;
$i$ := $0$;
while (true) {
  if ($n$ > $|\verB_i|$) {
    $i$ := $i+1$;
  }

  $\verA_n$ := $\verB_i[n]$;
  $n$ := $n+1$;
}
\end{lstlisting}
\end{minipage}\smallskip

Clearly, Property~\ref{it:cor-ch-len} in the definition of a~correct chain is guaranteed. Let $i\in\w$ and $n=|\verB_i|$. By the fact that $\verB_{i+1}\inx \verB_i$ we know that $\verB_{i+1}[n+1]\inxeq \verB_i \dL$. Therefore, for every $n\in\w$ we have $\verA_{n+1}\inxeq \verA_{n} \dR$ and if $n=|\verB_i|$ for some $i$ then $\verA_{n+1}\inxeq \verA_{n} \dL$. It implies that the sequence $(\verA_n)_{n\in\w}$ satisfies Property~\ref{it:cor-ch-ord} in the definition of a~correct chain. It is clearly witnessing for $X$ because it contains $(\verB_n)_{n\in\w}$ as a~subsequence.
\end{proof}

\subsection{Hardness for one counter}
\label{ssec:hard-one}

We are now ready to provide a definition of a $1$-blind counter B\"uchi automaton $\Aa_1$ recognising $\asigma{1}$-complete $\omega$-language. The automaton $\Aa_1$ is depicted in Figure~\ref{fig:one-vass}, with the convention that an edge $q\trans{a:j}q'$ represents a transition from the state $q$ to $q'$, over the letter $a$; that modifies the unique counter value by $j$, i.e.~$a\colon (q,c_1)\mapsto_{\Aa_1} (q',c_1+j)$. Moreover, the state $q_0$ is initial and $q_a$ is the only accepting state.

Let $\Sigma_0\eqdef\{{<}, d, {|}, i, {+}, {-}, {>}\}$ and consider the alphabet $\Sigma\eqdef \Sigma_0\cup\{\sharp\}$.

\medskip
An $\w$\=/word accepted by $\Aa_1$ consists of infinitely many \emph{phases} separated by $\sharp$. Each phase is a~finite word over the alphabet $\Sigma_0$. In our reduction we will restrict to phases being sequences of \emph{blocks}, each block being a~finite word of the form given by the following definition (for $n,m\in\w$ and $s\in\{{+},{-}\}$):
\begin{equation}
B^s(-n,+m)\eqdef {<}\ d^n\ {|}\ i^m\ s\ {>}\ \in A_0^\ast.
\label{eq:block-one}
\end{equation}

Such a~block is \emph{accepting} if $s={+}$, otherwise $s={-}$ and the block is \emph{rejecting}. If $\Aa_1$ starts reading a~block and moves from $q_0$ to $q_1$ over ${<}$ then we say that it \emph{chooses} this block. Otherwise $\Aa_1$ stays in $q_0$ and it does not \emph{choose} the given block. By the construction of the automaton $\Aa_1$, in every run it needs to choose exactly one block from each phase. Additionally, the run is accepting if and only if infinitely many of the chosen blocks are accepting.

\newcommand{\Myscale}{0.85}
\newcommand{\Trscale}{1.0}

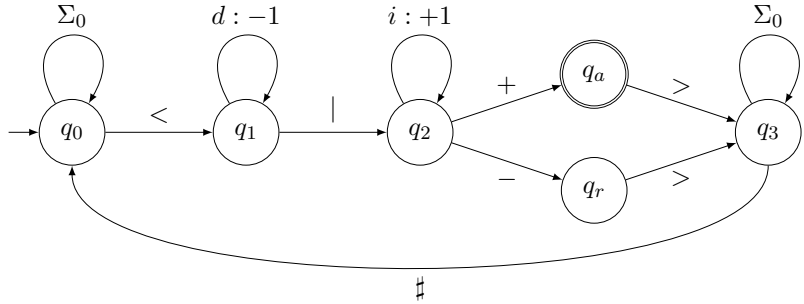
\begin{figure}[!t]
\vspace*{-2mm}
\centering
\scalebox{0.9}{
\begin{tikzpicture}[->,>=latex,auto,initial text={},scale=\Myscale]
\tikzstyle{trans}=[scale=\Trscale]
\node[state,initial  ] (q0) at (0,0) {$q_0$};
\node[state          ] (q1) at (3,0) {$q_1$};
\node[state          ] (q2) at (6,0) {$q_2$};
\node[state,accepting] (qf) at (9,+1) {$q_a$};
\node[state          ] (qr) at (9,-1) {$q_r$};
\node[state          ] (q3) at (12,0) {$q_3$};

\path (q0) edge [loop, in=60 , out=120, looseness=8] node[trans] {$\Sigma_0$} (q0)
           edge node[trans] {${<}$} (q1)
      (q1) edge [loop, in=60 , out=120, looseness=8] node[trans] {$d:-1$} (q1)
           edge node[trans] {${|}$} (q2)
      (q2) edge [loop, in=60 , out=120, looseness=8] node[trans] {$i:+1$} (q2)
           edge node[above, trans] {${+}$} (qf)
           edge node[below, trans] {${-}$} (qr)
      (qf) edge node[above, trans] {${>}$} (q3)
      (qr) edge node[below, trans] {${>}$} (q3)
      (q3) edge [loop, in=60 , out=120, looseness=8] node[trans] {$\Sigma_0$} (q3)
           edge [in=-90, out=270, looseness=0.5] node[below, trans] {$\sharp$} (q0);

\end{tikzpicture} }\vspace*{-2mm}
\caption{The automaton $\Aa_1$ with $1$-blind counter that recognises a~$\asigma{1}$\=/complete $\w$\=/language.}
\label{fig:one-vass}
\end{figure}

\begin{proposition}
\label{pro:hardness-one}
There exists a~continuous reduction from $\IFinx$ to the $\w$\=/language recognised by $\Aa_1$.
\end{proposition}

We will take a~set $X\subseteq \Tt$ and construct an~$\w$\=/word $\infA(X)$. This $\w$\=/word will be a~concatenation of infinitely many phases $\finA_0\sharp\finA_1\sharp\cdots$. The $n$\=/th phase $\finA_n$ will depend on $X\cap\{\dL,\dR\}^n$. The configurations $(q_0,c)$ reached at the beginning of an~$n$\=/th phase will be in correspondence with the nodes $\verA\in\{\dL,\dR\}^n$. The bigger the value $c$, the higher in the infix order (or the lexicographic order, as they overlap here) the respective node $\verA$ is.

\medskip
To precisely define our $\w$\=/word $\infA(X)$ we need to define some auxiliary functions. First, we define inductively the function $b\colon \Tt\to \w$, assigning to nodes $\verA\in\Tt$ their binary value $b(\verA)$:
\begin{itemize}
\itemsep=0.9pt
\item $b(\epsilon)=0$,
\item $b(\verA \dL) = 2\cdot b(\verA)$,
\item $b(\verA \dR) = 2\cdot b(\verA)+1$.
\end{itemize}

\noindent Note that for every $n\in\w$ we have
\[b\big(\{\dL,\dR\}^n\big)=\{0,1,\ldots,2^n{-}1\},\]
and the function is bijective between these sets.

\medskip
Now we can define fast\=/growing functions: $m\colon \{-1\}\cup\omega\to \w$ and $e\colon \Tt\to \w$:
\begin{align*}
m(-1)&= 1,\\
m(n)&= m(n-1)\cdot 2^n&\text{for $n\in\w$,}\\
e(\verA) &= m(|\verA|-1)\cdot b(\verA) &\text{for $\verA\in\Tt$.}
\end{align*}
These functions allow to use a~big range of the possible values of a~single counter of the automaton to represent particular nodes of the tree. We will use the following two invariants of this definition, for $n\in\w$ and $\verA,\verA'\in\{\dL,\dR\}^n$:\vspace*{-1mm}
\begin{align}
\verA\inx \verA' &\Longleftrightarrow e(\verA)\leq e(\verA'),\label{eq:order}\\
e(\verA) + m(|\verA|-1) &\ \,\leq\ \,m(|\verA|).
\label{eq:upper}
\end{align}

We take any $n=0,1,\ldots$ and define the $n$\=/th phase $\finA_n$. Let $\finA_n$ be the concatenation of the following blocks, for all $\verA\in\{\dL,\dR\}^n$ and $d\in\{\dL,\dR\}$:
\[B^s\big({-}e(\verA), {+}e(\verA d)\big),\]
where $s={+}$ if $\verA\in X$ and $d=\dL$; otherwise $s={-}$. Thus, the $n$\=/th phase is a~concatenation of $2^{n+1}$ blocks, one for each node $\verA d$ in $\{\dL,\dR\}^{n+1}$.

\smallskip
To conclude the proof of Proposition~\ref{pro:hardness-one} it is enough to prove the following two lemmas.

\begin{lemma}
\label{lem:chain-to-run}
If there exists a~correct chain witnessing for $X$ then $\infA(X)\in\lang(\Aa_1)$.
\end{lemma}

\begin{proof}
Consider a~correct chain $(\verA_n)_{n\in\w}$ witnessing for $X$. Assume that $I\subseteq \w$ is an~infinite set such that for $n\in I$ we have $\verA_n\in X$ and $\verA_{n+1}\inxeq \verA_n \dL$. Let us construct inductively a~run $\runA$ of $\Aa_1$ on $\infA(X)$. The invariant is that for each $n\in\w$ the configuration of $\runA$ before reading the $n$\=/th phase of $\infA(X)$ is of the form $(q_0,c_n)$ with $c_n\geq e(\verA_n)$. To define $\runA$ it is enough to decide which block to choose from an~$n$\=/th phase of $\infA(X)$:
\begin{itemize}
\item if $n\in I$ then choose the block $B^{+}\big({-}e(\verA_n), {+}e(\verA_n \dL)\big)$,
\item otherwise choose the block $B^{-}\big({-}e(\verA_n), {+}e(\verA_n \dR)\big)$.
\end{itemize}
Notice that by the invariant, it is allowed to choose the respective blocks as $c_n\geq e(\verA_n)$. Because of~\eqref{eq:order} and the fact that $(\verA_n)_{n\in\w}$ is a~correct chain, the invariant is preserved. As the set $I$ is infinite, the constructed run chooses an~accepting block infinitely many times and thus is accepting.
\end{proof}

\begin{lemma}
\label{lem:run-to-chain}
If $\infA(X)\in\lang(\Aa_1)$ then there exists a~correct chain witnessing for $X$.
\end{lemma}

\begin{proof}
Assume that $\runA$ is an~accepting run of $\Aa_1$ over $\infA(X)$. For $n=0,1,\ldots$ let $(q_0,c_n)$ be the configuration in $\runA$ before reading the $n$\=/th phase of $\infA(X)$ and assume that $\runA$ chooses a~block of the form $B^{s_n}\big({-}e(\verA_n), {+}e(\verA_n d_n)\big)$ in the $n$\=/th phase of $\infA(X)$. Our aim is to show that $(\verA_n)_{n\in\w}$ is a~correct chain witnessing for $X$. First notice that by the construction of $\infA(X)$ we have $|\verA_n|=n$.

\medskip
Clearly, as the counter needs to be non\=/negative, we have $e(\verA_n)\leq c_n$. Notice that by~\eqref{eq:upper} we obtain inductively for $n=0,1,\ldots$ that $c_n< m(n)$. Therefore, we have
\begin{align*}
m(n)\cdot b(\verA_{n+1}) = e(\verA_{n+1}) &\leq c_{n+1} =\\
= c_n - e(\verA_n) + e(\verA_n d_n) &< m(n)+e(\verA_nd_n) =\\
&=m(n)+m(n)\cdot b(\verA_n d_n).
\end{align*}
By dividing by $m(n)$ we obtain $b(\verA_{n+1})< 1+b(\verA_nd_n)$, thus $b(\verA_{n+1})\leq b(\verA_nd_n)$ and therefore $\verA_{n+1}\inxeq \verA_nd_n\inxeq \verA_n\dR$. Moreover, if $s_n={+}$ (i.e.~the $n$\=/th chosen block is accepting) then $\verA_n\in X$ and $d_n=\dL$. Therefore, as $\runA$ chooses infinitely many accepting blocks, $(\verA_n)_{n\in\w}$ is witnessing for $X$.
\end{proof}

This concludes the proof of Proposition~\ref{pro:hardness-one}. Thus the $\om$-language of $\Aa_1$ is indeed $\asigma{1}$-hard and therefore $\asigma{1}$\=/complete.

\subsection{Undecidable properties of $1$-blind counter $\om$-languages}
From the above  result we can now easily infer the following undecidability result.

\begin{theorem}\label{undec-borel}
It is undecidable to determine whether the $\om$-language of a  given $1$-blind counter automaton is Borel (respectively, in the Borel class $\bsigma{\alpha}$, in the Borel class
$\bpi{\alpha}$, for a given ordinal $\alpha$).
\end{theorem}

\begin{proof}
Let $L_1=\lang(\mathcal{A}_1)\subseteq \Sio$ be the  $\asigma{1}$-complete $\om$-language accepted by the $1$-blind counter automaton $\mathcal{A}_1$ given above.

\medskip
We consider the shuffle operation for two $\w$\=/words x and y  in $\Si^\omega$   given by
$${\rm Sh}(x,y) =x(1)y(1)x(2)y(2) ..  \in \Si^\omega$$
This is extended to the shuffle of $\omega$-languages by ${\rm Sh}(L,L')=\{ {\rm Sh}(x,y)  \mid  x \in L \mbox{ and } y \in L' \}$.

Let now $L$  be an $\omega$-language of a given $1$-blind counter automaton $\mathcal{A}$ over the alphabet $\Si$. We set  $S=  {\rm Sh}(L_1,  \Si^\omega)  \cup   {\rm Sh}( \Si^\omega, L)$
It is easy to see that one can construct, from   $\mathcal{A}$ and   $\mathcal{A}_1$,  another $1$-blind counter automaton $\mathcal{B}$  accepting $S$.

 There are now two cases:

\medskip
{\bf  First Case.} $  L= \Si^\omega$.

In that case $S= \Si^\omega  $ is in every Borel class  (and actually in every Wadge class, except the class of the empty set).

\medskip
{\bf   Second Case.} $L$  is not equal to   $\Si^\omega$.

In that case there exists an $\omega$\=/word $x\in \Si^\omega$ which is not in $L$.  Let now $T$ be the intersection of  $S$ and of  ${\rm Sh}( \Si^\omega, \{x\})$.  Then $T= {\rm Sh}(L_1,  \{x\})$ is also  ${\bf \Sigma}^1_1$-complete, and  since $ L_1$ is continuously reducible to
$S$ by $F: y \rightarrow {\rm Sh}( y, x)$  it follows that $S$ is  ${\bf \Sigma}^1_1$-complete.

Now the conclusion  follows from a recent result  that the universality problem for one blind counter B\"uchi automata is undecidable, see~\cite{BohmGHH17}.
\end{proof}

Notice that one can also get  other undecidability results, using the above one about topological properties.  First we can state the following theorem showing that the arithmetical properties of
 $1$-blind counter $\om$-languages are also undecidable.

\begin{theorem}\label{undec-effective-borel}
It is undecidable to determine whether the $\om$-language of a  given $1$-blind counter automaton is an effective $\Delta_1^1$-set (respectively, an arithmetical  $\Sigma^0_n$-set,  an arithmetical $\Pi^0_n$-set, for a given integer $n\geq 1$).
\end{theorem}

\begin{proof}
We can use the above proof of Theorem \ref{undec-borel}. Indeed,  in the first case the  $\omega$-language $S=\Si^\omega$ is in every arithmetical class. Moreover, in the second case the
$\omega$-language $S$ is  not a Borel set, and thus it is not an effective $\Delta_1^1$-set and does not belong to any arithmetical class.
\end{proof}

\begin{remark}
It is open to determine the exact complexity of these undecidable problems. In particular we do not know whether they are highly undecidable, as in the general case of Petri nets or of $4$-blind counter automata.
\end{remark}

We can also use topological properties to prove other undecidability properties which are not directly linked to topology.

\begin{theorem}\label{undec-other-results}
It is undecidable to determine whether the $\om$-language of a  given $1$-blind counter automaton $\mathcal{A}$:
\begin{enumerate}
\item is  a regular $\om$-language.
\item is accepted by a deterministic Petri net.
\item is accepted by a deterministic Turing machine.
\item has a complement $\lang(\mathcal{A})^-$ which is accepted by a Petri net.
\item has a complement $\lang(\mathcal{A})^-$ which is accepted by a Turing machine.
\end{enumerate}
\end{theorem}

\begin{proof}
We can again use the above proof of Theorem \ref{undec-borel}. Indeed,  in the first case the  $\omega$-language $S=\Si^\omega$ satisfies the five items of the theorem.  In the second case the $\omega$-language $S$ is non-Borel hence it is not a Boolean combination of ${\bf \Pi}^0_2$-sets and it does not satisfy any of the three first  items. Moreover, in this case the complement of $S$ is ${\bf \Pi}^1_1$-complete hence it cannot be accepted by any Turing machine and in particular by any Petri net (with B\"uchi acceptance condition).
\end{proof}


\section{Inherent non-determinism}
\label{sec:inherent-nondet}

In this section we formally state and prove the following corollary.

\begin{corollary}
\label{cor:no-model}
No model of deterministic, unambiguous, nor even countably\=/unambiguous automata with countably many configurations and a~Borel acceptance condition can capture the class of $\om$\=/languages recognisable by real\=/time $1$-blind counter B\"uchi automata.
\end{corollary}

It is expressed in the same spirit as the corresponding Theorem~5.5 in~\cite{hummel_msouplus}: we consider an~abstract model of automata $\Aa$ with a~countable set of configurations $C$, an initial configuration $c_\init\in C$, a~transition relation $\delta\subseteq C\times \Sigma\times C$, and an~acceptance condition $W\subseteq C^\w$. The notions of a~run $\run(\alpha,\rho)$; an~accepting run $\acc(\rho)$; and the language $\lang(\Aa)$ are defined in the standard way. Thus, under the assumption that the acceptance condition $W$ is Borel, the set
\[P\eqdef \big\{(\alpha,\rho)\in \Si^\w\times C^\w\mid \run(\alpha,\rho)\wedge \acc(\rho)\big\},\]
as in Definition~\ref{def:buchi-k-counter} is also Borel. The assumptions that the machine is deterministic, unambiguous, or countably\=/unambiguous imply that the cardinality of the sections $P_\alpha\eqdef \{\rho\mid (\alpha,\rho)\in P\}$ for $\alpha\in \Si^\w$ is at most countable. Therefore, the following \emph{small section theorem} by Lusin and Novikov applies.

\begin{theorem}[{see~\cite[Theorem~18.10]{Kechris94}}]
Let $X$, $Y$ be standard Borel spaces and let $P \subseteq X\times Y$ be Borel. If every section $P_x$ is countable, then $P$ has a~Borel uniformisation and therefore $\pi_X(P)$ is Borel.
\end{theorem}

Therefore, we know that $\lang(\Aa)=\pi_{\Si^\omega}(P)$ is Borel. Thus, no such machine can recognise $\lang(\Aa_1)$ for the automaton $\Aa_1$ from Section~\ref{sec:non-borel}, or any non-Borel $\om$-language of Petri nets obtained in Section \ref{sec:wadge} as these languages are  non\=/Borel.

\bigskip
Notice that the above Theorem of  Lusin and Novikov had already been used in the study of ambiguity of  context-free $\om$-languages in \cite{Fink-Sim} or of $\om$-languages of Turing machines in \cite{Fin-ambTM} and even  for tree languages  of tree automata \cite{FinkelS09}.  In particular,  it is proved in \cite{Fin-ambTM}  that if $L \subseteq  X^\om$ is accepted by a
B\"uchi Turing machine $\mathcal{T}$ and  $L$
is an analytic  but non-Borel set, then the set of $\om$-words,
which have $2^{\aleph_0}$ accepting runs by $\mathcal{T}$, has cardinality $2^{\aleph_0}$. This extends a similar result of  \cite{Fink-Sim} in the
case of context-free $\om$-languages and infinitary rational relations.
In that case we say that the  $\om$-language
$L$ has the maximum degree of ambiguity (with regard to acceptance by  B\"uchi Turing machines). From this result we can also infer the following one about $\om$-languages of Petri nets.

\begin{theorem}
Let $L \subseteq  \Sio$ be an $\om$-language accepted by a  B\"uchi $k$-counter automaton $\mathcal{A}$
   such that $L$
is an analytic but non-Borel set. The set of $\om$-words,
which have $2^{\aleph_0}$ accepting runs by $\mathcal{A}$, has cardinality $2^{\aleph_0}$.

\end{theorem}

Moreover, it is proved in \cite[Theorem 4.12]{Fin-ambTM}  that it is consistent with {\bf ZFC} that there exists an $\om$-language accepted by a real-time $1$-counter B\"uchi automaton
which belongs to the Borel class ${\bf \Pi}^0_2$ and which has the maximum degree of ambiguity with regard to acceptance by Turing machines. It is then easy to infer from this result and from the previous constructions of Section \ref{sec:wadge} that a similar result holds for  an $\om$-language in the Borel class ${\bf \Pi}^0_2$ accepted by a  $4$-blind counter B\"uchi automaton.

We end this section with the following undecidability result.

\begin{theorem}\label{undec-unambiguity}
It is undecidable to determine whether the $\om$-language of a  given $1$-blind counter automaton $\mathcal{A}$:
\begin{enumerate}

\item is accepted by an unambiguous Petri net.
\item has the maximum degree of ambiguity with regard to acceptance by Petri nets.
\item has the maximum degree of ambiguity with regard to acceptance by Turing machines.
\end{enumerate}
\end{theorem}

\begin{proof}
We can again use the above proof of Theorem \ref{undec-borel}. Indeed,  in the first case the  $\omega$-language $S=\Si^\omega$ is accepted by an  unambiguous (and even deterministic) Petri net (without any counter).
In the second case the $\omega$-language $S$ is non-Borel  and thus it has  the maximum degree of ambiguity with regard to acceptance by Petri nets or even by Turing machines.
\end{proof}

\section{Determinisation of unambiguous Petri nets}
\label{sec:determinisation}

The previous sections showed that non\=/deterministic blind counter automata are stronger in expressive power than any reasonable model of computation with a restricted form of non\=/determinism. This opens the question what is the actual expressive power of unambiguous blind counter automata. In this section we provide a construction allowing to simulate them using a variant of {\it deterministic } counter automata with \emph{copying}.

A counter automaton $\Mm=\langle K, \Si,$ $\Delta, q_0\rangle$ allows \emph{copying} if its transitions can additionally require to copy the value of one counter $\Cc_j$ to another counter $\Cc_{j'}$, symbolically $\Cc_{j'}:= \Cc_j$. The copying instructions can be represented by another component of the transition relation taken from the set $2^{(\{1,\ldots,k\}^2)}$ indicating which counters should be copied to which (we can assume that all the copying is executed simultaneously). A run of such a machine is defined analogously as in Section~\ref{sec:basic} with the additional requirement that the copying instructions are executed in the natural way.

\begin{theorem}
\label{thm:main-unamb}
If $\Aa$ is an~unambiguous blind counter B\"uchi automaton then $\lang(\Aa)$ can be recognised by a~deterministic Muller counter machine $\Mm$ that allows copying, hence by a deterministic Muller Turing machine. Moreover, the translation from $\Aa$ to $\Mm$ is effective.
\end{theorem}

The machine $\Mm$ above is not blind and during the construction we extensively use its ability to perform zero tests. For a discussion of the possible variants of the machine models involved in that theorem, see the end of this section.

It is known that every $\om$-language accepted by a deterministic Muller Turing machine is a Boolean combination of  arithmetical $\Pi^0_2$-sets, hence an arithmetical $\Delta^0_3$-set.  In particular,  an
$\om$-language accepted by a deterministic Muller Turing machine  is a~Boolean combination of $\bpi{2}$\=/sets, hence a~Borel $\bdelta{3}$\=/set \cite{Staiger97}.  Thus we can also state the following corollary of Theorem \ref{thm:main-unamb}.

\begin{corollary}
Assume that $\Aa$ is an unambiguous blind counter B\"uchi automaton. Then the  $\om$\=/lan\-guage $\lang(\Aa)$ is a~Boolean combination of $\bpi{2}$\=/sets, hence a~Borel $\bdelta{3}$\=/set. It is actually an~effective $\bdelta{3}$\=/set, i.e.~an~arithmetical (lightface) $\Delta^0_3$\=/set.
\end{corollary}

The overall structure of the construction is a variant of the powerset construction with an additional trimming. First we show how to split the space of configurations of a given unambiguous machine into finitely many regions (called clubs). Then we argue that the assumption of unambiguity implies, that after reading a prefix $w$ of the input word $\sigma$, the machine cannot reach two distinct configurations from the same club --- it would lead to two distinct accepting runs on a certain word $\sigma'\sqsupseteq w$. This means that it is enough to keep track of at most one run of the machine in each of the finitely many clubs. Based on that, we build a deterministic counter machine that stores all those finitely many runs in its memory.

\subsection{Lasso patterns}

We begin by recalling known structural properties of blind counter automata, namely \emph{lasso patterns} that are used to \emph{pump} runs of the automaton. The patterns lie at the core of the decidability algorithms for these automata. We will use these properties later to construct accepting runs of the machine under certain assumptions, which leads to the effectiveness of the provided translation.

Fix a~B\"uchi blind counter automaton $\Aa=\langle K,\Si,\Delta, q_0,F\rangle$ with a~set of states $K$, $k$\=/counters $\Cc_1,\ldots,\Cc_k$, and a transition relation $\Delta$. To avoid double indexing, we will denote a transition $(q,a,i_1$, $\ldots$, $i_k,q',j_1,\ldots,j_i)$ of $\Aa$ by $(q,a,I,q',J)$ with $I=(i_1,\ldots,i_k)\in \{0, 1\}^k$ and $J=(j_1,\ldots,j_k)\in \{-1, 0, 1\}^k$. Similarly, a configuration $(q,c_1,\ldots,c_k)$ of such a counter automaton can be written $(q,\tau)$ with $\tau=(c_1,\ldots,c_k)\in \N^k$.

\medskip
A~\emph{lasso pattern} is a~sequence of transitions $\big(\delta_i{=}(q_i,a_i,I_i,q'_i,J_i)\big)_{i=0,\ldots,\ell}\subseteq\Delta$ and a~number $0\leq \ell'\leq \ell$, such that the following conditions hold:
\begin{enumerate}
\itemsep=0.9pt
\item $q_0$ is the initial state of $\Aa$, and $I_0=(0, 0, \ldots , 0)$,
\item for each $i=0,1,\ldots,\ell{-}1$ we have $q'_i=q_{i+1}$,
\item $q'_\ell=q_{\ell'}$,
\item for each $i =0,1,\ldots,\ell$ the sum $\sum_{j=0}^{i} J_j$ belongs to $\N^k$ (i.e.~is coordinate\=/wise non\=/negative),
\item the sum $\sum_{j=\ell'}^\ell J_j$ also belongs to $\N^k$,
\item for each $i=0,1,\ldots,\ell$ and $c=1,\ldots,k$ if $(I_i)_c=1$ then $\sum_{j=0}^{i-1} (J_j)_c>0$.
\end{enumerate}
These conditions are meant to ensure that one can construct a run of $\Aa$ that uses consecutively the transitions from the lasso pattern. Notice that if $(I_i)_c=0$ in the last item of the definition then the assumption of blindness guarantees that there is another transition in $\Delta$ with $(I_i)_c=1$, so we don't need to restrict the sum $\sum_{j=0}^{i-1} (J_j)_c$ in that case.

We call $\ell$ the \emph{length} of the lasso pattern and $\ell'$ is its \emph{looping point}. A~lasso pattern as above is \emph{accepting} if for some $i=\ell',\ldots,\ell$ the state $q_i$ is accepting. The definition of a~lasso pattern is based on Problem~3.2 in~\cite{habermehl_petri_upper}. The exact properties used in the definition of a lasso pattern are chosen in such a way to ensure the following remark.

\begin{remark}
If a~blind counter B\"uchi automaton $\Aa$ has an~accepting lasso pattern $(\delta_i)_{i=0,\ldots,\ell}\subseteq\Delta$ with a looping point $0\leq \ell'\leq \ell$ then $\lang(\Aa)\neq\emptyset$.
\end{remark}

\begin{proof}
Let $u=a_0a_1\cdots a_{\ell'-1}$ and $w=a_{\ell'}a_{\ell'+1}\cdots a_{\ell}$. Then the $\w$\=/word $\alpha=uwww\cdots$ belongs to $\lang(\Aa)$ because one can construct an accepting run of $\Aa$ over this $\w$\=/word using the transitions of the assumed lasso pattern.
\end{proof}

The following theorem from~\cite[Section~3]{habermehl_petri_upper} implies decidability of the emptiness problem for blind counter B\"uchi automata by providing the converse implication.

\begin{theorem}
\label{thm:decidable}
If $\Aa$ is a~blind counter B\"uchi automaton such that $\lang(\Aa)\neq\emptyset$ then $\Aa$ has an~accepting lasso pattern of length bounded by a~function computable\footnote{The function is doubly\=/exponential, see the comment before Theorem~3.1 in~\cite{habermehl_petri_upper}.} based on~$\Aa$. As a~consequence, it is decidable if $\lang\big(\Aa\big)$ is empty.
\end{theorem}

\subsection{Clubs of configurations}

This section is devoted to an~introduction of a~technical concept used in the determinisation procedure: \emph{clubs of configurations}. These are regions of the configuration space of a counter automaton that represent somehow similar behaviour of the machine. We will see later on that certain clubs that are \emph{optimal} can be treated in a~homogeneous way (Lemma~\ref{lem:club-non-empty}); and moreover each club can be split into a finite family of optimal ones (Proposition~\ref{pro:split-clubs}).

\medskip
Fix a~counter automaton~$\Aa$ with a~set of states~$K$ and $k$\=/counters $\Cc_1,\ldots,\Cc_k$. Let $N\in\N$ and $\gamma=\big(\gamma_1,\ldots,\gamma_k\big)$ be a~vector where each $\gamma_i$ for $i=1,\ldots,k$ is either a~natural number or the expression $({\geq} N)$. Recall that we will denote the configurations of $\Aa$ by $(q,\tau)$ with $\tau=(c_1,\ldots,c_k)\in\N^k$ being the vector of counter values. A~\emph{club} is a~set of configurations of $\Aa$ of the form
\begin{equation}
[q,\gamma]\eqdef \big\{ (q,\tau)\mid \forall{1\leq i \leq k}.\ (\gamma_i=\tau_i\in\N) \lor (\gamma_i=({\geq} N) \land \tau_i\geq N) \big\}.
\end{equation}
The \emph{dimension} of a~club is the number of expressions $({\geq} N)$ that appear in~$\gamma$. Similarly, the value $N\in \N$ is the \emph{threshold} of the club. The \emph{minimal configuration} of a~club $[q,\gamma]$ is the configuration $(q,\tau)$ where for $i=1,\ldots,k$ the coordinate $\tau_i$ equals $\gamma_i$ when $\gamma_i\in\N$ and equals $N$ otherwise. Notice that the minimal configuration of a~club is the ${\preceq}$\=/least element of the club (see the definition of the simulation order on page~\pageref{def:sim-order} and Remark~\ref{rem:residual-order}).

\medskip
Let $[q,\gamma]$ be a~club with threshold $N$ and $M\geq N$ be a~natural number. By $[q,\gamma\restr_M]$ (called the \emph{restriction} of $[q,\gamma]$ to the threshold $M$) we denote the club obtained from $[q,\gamma]$ by replacing each occurrence of $({\geq}N)$ by $({\geq}M)$. Notice that, as sets of configurations we have
\begin{equation}
\label{eq:subset-clubs}
[q,\gamma\restr_M]\subseteq [q,\gamma].
\end{equation}

\subsection{Languages of clubs}

Each club $[q,\gamma]$, seen as a set of configurations of $\Aa$, induces its $\omega$\=/language $\lang\big(\Aa, [q,\gamma]\big)$ being just the set theoretic union of all the $\omega$\=/languages $\lang\big(\Aa,(q,\tau)\big)$ for all configurations $(q,\tau)\in[q,\gamma]$. Although the notation might suggest that, we do not consider the possibility to treat clubs as configurations of a counter automaton and perform transitions on clubs --- instead we execute the automaton from each of the single configurations $(q,\tau)\in[q,\gamma]$ and then take the union of these $\omega$-languages.

We will now show how to decide if the $\omega$-language of a club is empty, see Corollary~\ref{cor:decidable} at the end of this subsection. For that, fix a~$k$-blind counter B\"uchi automaton $\Aa$ and a~club $[q,\gamma]$ of dimension $d$ and threshold $N$.

Let $(q,\tau_0)$ be the minimal configuration of $[q,\gamma]$. Without loss of generality we can assume that $\gamma=\big(\gamma_1,\ldots,\gamma_{k'}, ({\geq} N), \ldots,({\geq} N)\big)$ with the values $\gamma_1,\ldots,\gamma_{k'}$ being natural numbers and $k={k'}+d$. Notice that by the choice of $\tau_0$ we know that $\tau_0=\big(\gamma_1,\ldots,\gamma_{k'}, N, \ldots, N\big)$.

Let $a_0,a_1,a_2\in \Sigma$ be any (not necessarily distinct) fixed letters of the alphabet $\Sigma$. Let $\tau_0^1,\ldots,\tau_0^Z$ be a sequence of vectors in $\{0,1\}^k$ such that $\sum_{j=1}^Z \tau_0^j = \tau_0$, with $Z\in\N$ being the maximal of the coordinates of $\tau_0$. Let $\tau_1=(0,\ldots,0,1,\ldots,1)$ be a~vector with $k'$ zeros followed by $d$ ones.

Consider the blind counter B\"uchi automaton denoted $\big([q,\gamma]\cdot\Aa\big)$ depicted on Figure~\ref{fig:club-test}. This automaton first reads a sequence of letters $a_0$ increasing all the counters to the exact values given by $\tau_0$ (using the vectors $\tau_0^j$ for that); then it can arbitrarily many times read $a_1$ and increase the counters numbered $k'{+}1,k'{+}2,\ldots,k$, i.e.~those corresponding to the value $(\geq N)$ in $\gamma$; and then it reads $a_2$ and moves to a~copy of the automaton $\Aa$ into the state $q$.

\begin{figure}[h]
\centering
\scalebox{1.1}{
\begin{tikzpicture}[->,>=latex,auto,initial text={},scale=0.9]
\tikzstyle{trans}=[scale=0.8]
\node[state,initial  ] (q0) at (0.0,0) {$q_0$};
\node[state          ] (q1) at (2.5,0) {$q_1$};
\node (qd) at (5.0,0) {$\cdots$};
\node[state          ] (qZ) at (7.5,0) {$q_Z$};

\node at (12,1.5) {$\Aa$};
\draw[rounded corners=10] (9, -1) rectangle (13, 2.5);

\node[state          ] (qf) at (10,0) {$q$};

\path (q0) edge node[trans] {$a_0:\tau_0^1$} (q1);
\path (q1) edge node[trans] {$a_0:\tau_0^2$} (qd);
\path (qd) edge node[trans] {$a_0:\tau_0^Z$} (qZ);
\path (qZ) edge [loop, in=60 , out=120, looseness=8] node[trans] {$a_1:\tau_1$} (qZ)
           edge node[trans] {$a_2$} (qf);
\end{tikzpicture} }\vspace*{1mm}
\caption{The automaton denoted $\big([q,\gamma]\cdot\Aa\big)$ that checks non\=/emptiness of the set $\lang\big(\Aa,[q,\gamma]\big)$.}
\label{fig:club-test}
\end{figure}

\begin{proposition}
\label{pro:update-club}
Fix a~blind counter B\"uchi automaton $\Aa$ and a~club $[q,\gamma]$ with threshold $N$. Then one can effectively compute a~number $M\geq N$ such that the following conditions are equivalent:
\begin{enumerate}
\item $\lang\big(\Aa,[q,\gamma]\big)\neq\emptyset$,
\item $\lang\big([q,\gamma]\cdot\Aa\big)\neq\emptyset$,
\item $\lang\big(\Aa,(q,\tau'_0)\big)\neq\emptyset$ for $(q,\tau'_0)$ the minimal configuration of $[q,\gamma\restr_M]$,
\item $\lang\big(\Aa,[q,\gamma\restr_M]\big)\neq\emptyset$.
\end{enumerate}
\end{proposition}

\begin{proof}
Let us denote by $\Aa'$ the automaton $\big([q,\gamma]\cdot\Aa\big)$. Let $M_\ell$ be the bound on the length of a~lasso pattern for $\Aa'$ given by Theorem~\ref{thm:decidable} and take $M\eqdef N+M_\ell$.

We start with the implication $1\Rightarrow 2$. Assume that $\alpha\in \lang\big(\Aa,(q,\tau)\big)$ for some $(q,\tau)\in [q,\gamma]$. Let $M'$ be the maximal coordinate of $\tau$. Consider $\alpha'=a_0^Z a_1^{M'} a_2\alpha$, where $Z$ is taken as in the construction of $\Aa'$. Let $(q,\tau')$ be the configuration of $\Aa'$ reached after reading the prefix $a_0^Z a_1^{M'} a_2$ of $\alpha'$. By the choice of $M'$, we know that $\tau'\succeq \tau$. Therefore, by Remark~\ref{rem:residual-order} we have $\alpha\in\lang\big(\Aa,(q,\tau')\big)$. Thus, there exists an~accepting run of $\Aa'$ over $\alpha'$ and $\lang\big(\Aa'\big)\neq\emptyset$.

Now consider the implication $2\Rightarrow 3$. Theorem~\ref{thm:decidable} together with the choice of $M_\ell$ guarantee that if $\lang\big(\Aa'\big)\neq\emptyset$ then there exists an accepting lasso pattern of $\Aa'$. Let $\big(\delta_i{=}(q_i,a_i,I_i,q'_i,J_i)\big)_{i=0,\ldots,\ell}\subseteq\Delta$ and $0\leq\ell'\leq \ell$ be such a~lasso pattern.

As the states $q_0,\ldots, q_Z$ of $\Aa'$ are not accepting, it means that the state $q$ must appear among $(q_i)_{i\leq \ell}$. Let $j$ be the minimal index such that $q_j=q$. Since $q_0,\ldots,q_Z$ are not reachable from the copy of $\Aa$ within $\Aa'$, we know that $Z< j\leq \ell'$. Consider the run $\runA$ of $\Aa'$ over an $\w$\=/word $\alpha$ that is obtained by following the transitions of the considered lasso pattern. We know that $\runA(j)=(q,\tau)$ for a~certain configuration $(q,\tau)$ and the considered state $q$. The rest of this run is accepting, witnessing that $\lang\big(\Aa,(q,\tau)\big)\neq\emptyset$. However, by the construction of $\Aa'$ we know that $\tau=\tau_0+\tau_1\cdot (j{-}Z{-}1)$.

Recall that $M=N+M_\ell$, which implies that $\tau_0'$ from Item~3 of the statement has the form $\tau_0+\tau_1\cdot M_\ell$. As $j\leq M_\ell$, we know that $(q,\tau)\preceq (q,\tau_0')$ and Remark~\ref{rem:residual-order} implies that $\lang\big(\Aa,(q,\tau'_0)\big)\neq\emptyset$.

The implication $3\Rightarrow 4$ is obvious, as $(q,\tau'_0)\in [q,\gamma\restr_M]$. Similarly, $4\Rightarrow 1$ is also clear because $[q,\gamma\restr_M]\subseteq [q,\gamma]$, see~\eqref{eq:subset-clubs}.
\end{proof}

\begin{corollary}
\label{cor:decidable}
It is decidable for a~club $[q,\gamma]$ of $\Aa$ if the set of $\omega$\=/words $\lang\big(\Aa,[q,\gamma]\big)$ is empty.
\end{corollary}

\subsection{Optimal clubs}

A~club $[q,\gamma]$ is called \emph{trivial} if $\lang\big(\Aa,[q,\gamma]\big)=\emptyset$; otherwise $[q,\gamma]$ is called \emph{non trivial}. Proposition~\ref{pro:update-club} implies that each non-trivial club $[q,\gamma]$ can be restricted to another non-trivial club $[q,\gamma\restr_M]\subseteq [q,\gamma]$ such that already the $\omega$-language $\lang\big(\Aa,(q,\tau'_0)\big)$ is non empty for $\tau'_0$ the minimal configuration of $[q,\gamma\restr_M]$. In the latter part of the construction we will be interested in such \emph{optimal} clubs.

Let $[q,\gamma]$ be a~club with the minimal configuration $(q,\tau)$. Then $[q,\gamma]$ is called \emph{optimal} if the following implication holds:
\[\lang\big(\Aa,[q,\gamma]\big)\neq\emptyset\ \Longrightarrow\ \lang\big(\Aa,(q,\tau)\big)\neq\emptyset.\]
Notice that each club of dimension $0$, as a~set of configurations, is a~singleton and therefore it is optimal by the definition. Obviously, an~optimal non\=/trivial club with the minimal configuration $(q,\tau)$ satisfies $\lang\big(\Aa,(q,\tau)\big)\neq\emptyset$.

\begin{remark}
By Theorem~\ref{thm:decidable} and Corollary~\ref{cor:decidable} it is decidable whether a~given club is trivial and whether it is optimal.
\end{remark}

\begin{lemma}
\label{lem:subset-optimal}
If $[q,\gamma]\supseteq [q,\gamma']$ are two clubs and $[q,\gamma]$ is optimal then also $[q,\gamma']$ is optimal.
\end{lemma}

\begin{proof}
By the assumption that $[q,\gamma]\supseteq [q,\gamma']$, the minimal configuration $(q,\tau)$ of $[q,\gamma]$ is ${\preceq}$-smaller than the minimal configuration $(q,\tau')$ of $[q,\gamma']$. Thus, Remark~\ref{rem:residual-order} implies that if the $\omega$-language $\lang\big(\Aa,(q,\tau)\big)$ is non empty then also $\lang\big(\Aa,(q,\tau')\big)$ is non empty.
\end{proof}

The following lemma shows that each club can be made optimal by increasing its threshold.

\begin{lemma}
\label{lem:club-non-empty}
If $[q,\gamma]$ is a~club with threshold $N$ then there exists $M\geq N$ such that the club $[q,\gamma\restr_M]$ is optimal. Moreover, the value of $M$ can be effectively computed based on $\Aa$ and a~representation of $[q,\gamma]$.
\end{lemma}

\begin{proof}
It is enough to take the value $M$ from Proposition~\ref{pro:update-club}. The implication $4\Rightarrow 3$ of the proposition implies that $[q,\gamma\restr_M]$ is always optimal.
\end{proof}

If one doesn't care about the computability of the value $M$ in Lemma~\ref{lem:club-non-empty}, then one can use directly Remark~\ref{rem:residual-order}, as expressed by the following remark.

\begin{remark}
\label{rem:subclub-noneff}
Every machine model satisfying Remark~\ref{rem:residual-order} has the following property:
if $[q,\gamma]$ is a~club with threshold $N$ then there exists $M\geq N$ such that the club $[q,\gamma\restr_M]$ is optimal.
\end{remark}

\begin{proof}
Let $\Aa$ be a machine of the considered model that satisfies Remark~\ref{rem:residual-order}. If $\lang\big(\Aa,[q,\gamma]\big)$ is empty then the club is already optimal for $M=N$. Otherwise, let $(q,\tau)\in[q,\gamma]$ be a configuration such that $\lang\big(\Aa,(q,\tau)\big)\neq\emptyset$. Let $M$ be the maximal coordinate of $\tau$ and let $(q,\tau_0)$ be the minimal configuration of $[q,\gamma\restr_M]$ (it is obtained from $\gamma$ by replacing each coordinate $({\geq}N)$ by $M$). Then $(q,\tau)\preceq (q,\tau_0)$ and therefore $\lang\big(\Aa,(q,\tau_0)\big)\neq\emptyset$, which means that the club $[q,\gamma\restr_M]$ is optimal.
\end{proof}

The following proposition is the technical core of the construction. We believe that it is of independent interest.

\begin{proposition}
\label{pro:split-clubs}
Let $[q,\gamma]$ be a~club. Then $[q,\gamma]$, as a~set of configurations, can be written as a~finite pairwise disjoint union of optimal clubs. Moreover, such a~decomposition can be computed effectively.
\end{proposition}

\begin{proof}
The proof is inductive on the dimension $d$ of $[q,\gamma]$. Since each club of dimension $0$ is optimal, the thesis holds for $d=0$. Assume the thesis for all the clubs of dimensions at most $d$ and consider a~club $[q,\gamma]$ of dimension $d>0$ with threshold $N$. Without loss of generality we can assume that $\gamma=\big(\gamma_1,\ldots,\gamma_{k'}, ({\geq} N), \ldots, ({\geq} N)\big)$ --- the total number $k$ of coordinates of $\gamma$ is $k'{+}d$ and all the values $\gamma_1,\ldots,\gamma_{k'}$ are natural numbers.

Apply Lemma~\ref{lem:club-non-empty} to $[q,\gamma]$ to obtain $M\geq N$ such that the restriction $[q,\gamma\restr_M]$ is optimal. Notice that $\gamma\restr_M=\big(\gamma_1,\ldots,\gamma_{k'}, ({\geq} M), \ldots, ({\geq} M)\big)$.

Let $F$ be the set of clubs of the form $[q,\gamma']$ where $\gamma'$ equals $\gamma$ on coordinates $1$,\ldots, $k'$ and for each coordinate $i=k'{+1}$,\ldots, $k$ either $\gamma'_i=({\geq} M)$ or $\gamma'_i$ is a~natural number satisfying $N\leq \gamma'_i < M$. Since there is exactly $2^d$ choices for a~set of coordinates with $({\geq M})$ in $\gamma$, and for each such choice there is only finitely many clubs with $(\geq M)$ on exactly those coordinates in $F$, the set $F$ is finite.

\begin{claim}
The clubs in $F$ are pairwise disjoint and $[q,\gamma]=\bigcup F$.
\end{claim}

The disjointness follows directly from the construction. For the union, it is enough to notice that each $[q,\gamma']\in F$ satisfies $[q,\gamma']\subseteq [q,\gamma]$ and each $(q,\tau)\in [q,\gamma]$ can be found in one of the clubs of $F$.

\medskip
Notice that $[q,\gamma\restr_M]\in F$ --- it corresponds to the choice of all~$d$ coordinates being $({\geq}M)$. Let $F'=F\setminus\big\{[q,\gamma\restr_M]\big\}$. Observe that if $[q,\gamma']\in F'$ then the dimension of $[q,\gamma']$ is at most $d{-}1$. Thus, we can apply the inductive assumption to each club in $F'$ and take the union of all these clubs. Let $F''$ be the set of the clubs obtained this way. The clubs in $F''$ are pairwise disjoint, optimal, and disjoint from $[q,\gamma\restr_M]$. Thus,
\[[q,\gamma]=\bigcup F'' \cup [q,\gamma\restr_M],\]
is a decomposition of $[q,\gamma]$ into finitely many pairwise disjoint optimal clubs.

\medskip
For the effectiveness of the above construction, it is enough to observe that the bound $M$ given by Lemma~\ref{lem:club-non-empty} is effective and the rest of the construction is just a~recursive application of the same procedure.
\end{proof}

\subsection{Unambiguous automata}

Fix an~unambiguous blind counter B\"uchi automaton $\Aa$ with a~set of states~$K$ and $k$\=/blind counters $\Cc_1,\ldots,\Cc_k$. Our goal is to show that such an automaton cannot simultaneously (reading a finite word) reach two distinct configurations belonging to the same optimal non-trivial club.

\begin{lemma}
\label{lem:unique}
Let $[q,\gamma]$ be an~optimal club and $w\in \Sigma^\ast$ be a finite word. Assume that for $i=1,2$ there exists a~configuration $(q,\tau_i)\in[q,\gamma]$ that can be reached while reading $w$ from the initial configuration. If $\tau_1\neq \tau_2$ then the club $[q,\gamma]$ is trivial.

Similarly, if a configuration $(q,\tau)$ can be reached by two distinct runs while reading a finite word $w$ from the initial configuration then $\lang\big(\Aa,(q,\tau)\big)$ is empty.
\end{lemma}

\begin{proof}
Consider the first statement of the lemma. Assume to the contrary that $[q,\gamma]$ is non trivial. Let $(q,\tau)$ be the minimal configuration of $[q,\gamma]$. By the assumption of optimality, the $\omega$-language $\lang\big(\Aa,(q,\tau)\big)$ is non\=/empty, let $\alpha\in \lang\big(\Aa,(q,\tau)\big)$ be an $\omega$\=/word witnessing that. Since both $(q,\tau_1)$ and $(q,\tau_2)$ belong to $[q,\gamma]$, we know that $(q,\tau)\preceq (q,\tau'_1)$ and $(q,\tau)\preceq (q,\tau_2)$. Thus, Remark~\ref{rem:residual-order} implies that
\[\alpha\in \lang\big(\Aa,(q,\tau_1)\big)\cap \lang\big(\Aa,(q,\tau_2)\big).\]
This gives a~contradiction because it means that $\Aa$ has two distinct accepting runs over $w\cdot \alpha$.

For the second statement of the lemma, the proof is analogous: if $\alpha\in\lang\big(\Aa,(q,\tau)\big)$ then $\Aa$ has two distinct accepting runs over $w\cdot \alpha$.
\end{proof}

Notice that we can easily extend the above lemma to $p$-unambiguous blind counter B\"uchi automaton, for an integer $p \geq 1$.  Let us call a  blind counter B\"uchi automaton $p$-unambiguous if every $\om$-word over the input alphabet has at most $p$ accepting runs. Then we can state the following result which holds for a $p$-unambiguous blind counter B\"uchi automaton. The proof   is similar to the above one.

\begin{lemma}
\label{lem:unique2}
Let $[q,\gamma]$ be an~optimal club and $w\in \Sigma^\ast$ be a finite word. Assume that for $i=1,2, \ldots , p, p+1$ there exists a~configuration $(q,\tau_i)\in[q,\gamma]$ that can be reached while reading $w$ from the initial configuration. If for all $i, j \in [1, p+1]$ $i\neq j \rightarrow \tau_i\neq \tau_j$ then the club $[q,\gamma]$ is trivial.

Similarly, if a configuration $(q,\tau)$ can be reached by $p+1$  distinct runs while reading a finite word $w$ from the initial configuration then $\lang\big(\Aa,(q,\tau)\big)$ is empty.
\end{lemma}

\subsection{The machine $\Mm$}

The construction of the machine $\Mm$ is based on the above Lemma  \ref{lem:unique}. Thanks to it, we know that in the naive powerset construction, it is enough to remember at most one configuration of $\Aa$ for each optimal club. Moreover, by Proposition~\ref{pro:split-clubs}, one can split the whole configuration space of $\Aa$ into finitely many pairwise disjoint optimal clubs.

Let $F$ be a (finite) family of these clubs. Let the machine $\Mm$ store at most one configuration $(q,\tau)$ for each club $[q,\gamma]\in F$ (of course we require that $(q,\tau)\in[q,\gamma]$). Notice that a finite family of counters is enough to store all these configurations at once.

Upon reading a successive letter, $\Mm$ can update all the stored configurations according to all the possible transitions of $\Aa$. Whenever two distinct configurations obtained that way belong to the same club of $F$ (we call such a situation a \emph{collision}), $\Mm$ can \emph{discard} both these configurations, because the club is guaranteed to be trivial.

More precisely, there are in fact two possible scenarios for a configuration $(q',\tau')$ to be discarded because of a collision. The first case is that there might be a distinct configuration $(q',\tau'')\neq(q',\tau')$ reachable by one of the simulated transitions, such that both $(q',\tau')$ and $(q',\tau'')$ belong to a single club $[q',\gamma']\in F$. The second case is that the same configuration $(q',\tau')$ can also be reached by a different transition from one of the previously stored configurations. However, in both cases Lemma~\ref{lem:unique} applies, guaranteeing that $\lang\big(\Aa,(q',\tau')\big)=\emptyset$.

Due to the policy of discarding, the invariant of storing at most one configuration per club from $F$ is preserved.

Finally, $\Mm$ uses a Muller condition to check if any of the runs of $\Aa$ that are simulated in parallel, turned out to be accepting.

See Appendix~\ref{app:construction} for a precise construction of $\Mm$ and a discussion of exact computational features used in its construction.

This concludes the proof of Theorem~\ref{thm:main-unamb}.

\begin{remark}
Notice that Theorem \ref{thm:main-unamb} can be extended to the case of a $p$-unambiguous blind counter B\"uchi automaton, using Lemma \ref{lem:unique2} instead of  Lemma \ref{lem:unique}.
We do not enter into the details which are left to the reader. The ideas are identical but the machine  $\Mm$ will be simply more complicated since it has to store at most $p$  configurations $(q,\tau)$ for each club $[q,\gamma]\in F$.
\end{remark}

\begin{theorem}
\label{thm:main-unamb2}
If $\Aa$ is a $p$-unambiguous blind counter B\"uchi automaton, for an integer $p\geq 1$,  then $\lang(\Aa)$ can be recognised by a~deterministic Muller counter machine $\Mm$ that allows copying, hence by a deterministic Muller Turing machine. Moreover, the translation from $\Aa$ to $\Mm$ is effective.
\end{theorem}

\begin{corollary}
Assume that $\Aa$ is a $p$-unambiguous blind counter B\"uchi automaton,  for an integer $p\geq 1$. Then the $\om$\=/language $\lang(\Aa)$ is a~Boolean combination of $\bpi{2}$\=/sets, hence a~Borel $\bdelta{3}$\=/set. It is actually an~effective $\bdelta{3}$\=/set, i.e.~an~arithmetical (lightface) $\Delta^0_3$\=/set.
\end{corollary}

\section{Concluding remarks}
\label{sec:conclusions}

We have proved  that  the Wadge hierarchy of  Petri nets $\om$-languages, and even of  $\om$-languages in the class {\bf r}-${\bf BCL}(4)_\om$,
is equal to the Wadge hierarchy of  effective analytic sets, and that it is highly undecidable to determine the topological complexity of a Petri net $\om$-language.
Based on the constructions used in the proofs of the above results, we have also shown that the equivalence and  the inclusion problems for
$\om$-languages of Petri nets are $\Pi_2^1$-complete, hence highly undecidable.
In some sense, from the two points of view of the topological complexity and of highly undecidable problems, our results show that, in contrast with the finite behaviour,  the infinite behaviour of non-deterministic Petri nets is closer to the infinite behaviour of Turing machines than to that of finite automata.

As further developments showing the inherent complexity of the model, we have proved that the determinacy of Wadge games between two players in charge of $\om$-languages of Petri nets is equivalent to the (effective) analytic determinacy,  which is known to be a large cardinal assumption, and thus is not provable in the axiomatic system {\bf ZFC}. We have also provided a Petri net whose $\om$-language is either a Borel ${\bf \Pi}^0_2$-set or a non-Borel set, depending on the model of {\bf ZFC} under consideration.

Additionally, we have shown that in fact only one counter is enough to obtain a single $\asigma{1}$\=/complete $\omega$-language, i.e.~an $\omega$-language of maximal topological complexity among those recognisable by Petri nets. All these results imply that non-deterministic Petri nets are expressively stronger than any reasonable model of deterministic or unambiguous machines.

We have also studied the expressive power of unambiguous Petri nets. As it turns out, they admit a~determinisation procedure. As a consequence, the topological complexity of their $\omega$-languages is low in the Borel hierarchy (they all are $\bdelta{3}$ sets).

It remains open  for further study to determine the Borel and Wadge hierarchies of $\om$-languages accepted by automata with less than four blind counters. In particular, it then remains open to determine whether there exist some  $\om$-languages accepted by 1-blind-counter automata which are Borel of rank greater than $3$.

\paragraph*{Acknowledgements}

The authors would like to cordially thank Georg Zetzsche for his suggestions about decidability issues for blind counter automata.

The authors also thank anonymous referees of this paper and  of  the papers incorporated in this journal version --- their suggestions allowed to improve the presentation.


\appendix

\newpage

\section{Construction of the machine $\Mm$ from Section~\ref{sec:determinisation}}
\label{app:construction}

In this appendix we provide a precise construction of a real time counter machine $\Mm$ with zero tests and copying, as stated in Theorem~\ref{thm:main-unamb}. We assume that an~unambiguous blind counter B\"uchi automaton $\Aa$ is fixed.

\subsection{Simulating transitions}

First we show how it can represent single configurations of $\Aa$ and simulate transitions.

Let $[q,\gamma_0]=\big[q,\big(({\geq}0), ({\geq}0),\ldots,({\geq}0)\big)\big]$ be a~maximal~club of dimension $k$ (recall that $k$ is the number of blind counters of the input automaton $\Aa$). Notice that $\bigcup_{q\in K} [q,\gamma_0]$ is the set of all configurations of $\Aa$. Apply Proposition~\ref{pro:split-clubs} to each of the clubs $[q,\gamma_0]$ obtaining a~finite set of clubs~$F_q$. Let $F\eqdef \bigcup_{q\in K} F_q$. Then $F$ is a~finite set of pairwise disjoint optimal clubs and $\bigcup F$ is the set of all configurations of~$\Aa$. By further splitting the clubs in $F$ and applying Lemma~\ref{lem:subset-optimal}, we can ensure that all the clubs in $F$ share the same threshold $N$ --- it is enough to take as $N$ the maximal threshold of the clubs in $F$ and then perform a similar splitting as in the proof of Proposition~\ref{pro:split-clubs}. Clearly we can ensure that $N>0$.
Notice that the set $F$ can be effectively computed based on $\Aa$ --- it is enough to use Proposition~\ref{pro:split-clubs} for that.

For each club $[q,\gamma]\in F$ the machine $\Mm$ has a separate set of $k$ counters, denoted $\Cc_1^{[q,\gamma]},\ldots, \Cc_k^{[q,\gamma]}$. This means that $\Mm$ has $k\cdot |F|$ counters in total. A configuration $(q,\tau)\in [q,\gamma]\in F$ is represented in these counters by subtracting~$N$ from each of the counter values, i.e.~for $j=1,\ldots,k$ the counter $\Cc_j^{[q,\gamma]}$ stores the value $\max(\tau_j{-}N,0)$. Notice that if the $j$th coordinate of $\gamma$ equals $({\geq} N)$ then $\tau_j{-}N\geq 0$ and our way of storing the value is exact. On the other hand, if $\gamma_j$ is a natural number then $\tau_j=\gamma_j$, because $(q,\tau)\in [q,\gamma]$. This means that in this case the value $\tau_j$ (even if smaller than $N$) is determined by $\gamma$ and therefore known.

Now we will show how the machine $\Mm$ can simulate a transition $a: (q,\tau) \mapsto_{\Aa}  (q',\tau')$ of $\Aa$. Assume that $(q,\tau)\in [q,\gamma]\in F$ and the configuration $(q,\tau)$ is stored in the counters $\Cc_1^{[q,\gamma]},\ldots, \Cc_k^{[q,\gamma]}$ of~$\Mm$. First notice that based on $\gamma$ the machine $\Mm$ can decide if the transition is possible at all, i.e. if the non-negativity conditions of the transition are met by $\tau$ (the assumption that $N>0$ is used here). Moreover, for each coordinate $j$ such that $\gamma_j = ({\geq}N)$, the machine $\Mm$ can use a zero test on $\Cc_j^{[q,\gamma]}$ to determine if $\tau_j=N$ or $\tau_j>N$. The remaining coordinates of $\tau$ are fixed by the knowledge of $\gamma$. This allows $\Mm$ to determine the unique club $[q',\gamma']\in F$ such that $(q',\tau')\in [q',\gamma']$. Thus, $\Mm$ can copy the values from the counters $\Cc_1^{[q,\gamma]},\ldots, \Cc_k^{[q,\gamma]}$ to the counters $\Cc_1^{[q',\gamma']},\ldots, \Cc_k^{[q',\gamma']}$ and additionally perform the counter modifications as in the original transition: if $\tau'_j\neq \tau_j$ then the machine updates $\Cc_j^{[q',\gamma']}$ in such a way to ensure that $\Cc_j^{[q',\gamma']}=\max(\tau'_j{-}N,0)$ --- this update may also require to perform a zero test on $\Cc_j^{[q,\gamma]}$ to know if $\tau_j=N$ or not.

There are two technical subtleties of the above construction. First $[q',\gamma']$ might be equal to $[q,\gamma]$ and then no copying is needed. Second, it might be the case that $\gamma_j=\tau_j$ is a natural number greater than $N$ in which case the value of $\Cc_j^{[q,\gamma]}$ is a positive number equal to $\tau_j-N$. In that case that number needs to be stored in $\Cc_j^{[q,\gamma]}$ (even though it is fixed by $\gamma$) and then copied to $\Cc_j^{[q',\gamma']}$ because it might be the case that $\gamma'_j=({\geq}N)$.

We exploit here the fact that the thresholds of all the clubs in $F$ are the same.

\subsection{Simulating non-determinism}

The above simulation allows us to represent one configuration of $\Aa$ for each club $[q,\gamma]\in F$. Moreover, we know how to perform transitions on these configurations. We can also simulate non\=/determinism of $\Aa$: if $a: (q,\tau) \mapsto_{\Aa}  (q',\tau')$ and $a: (q,\tau) \mapsto_{\Aa}  (q'',\tau'')$ such that $(q',\tau')$ and $(q'',\tau'')$ belong to distinct clubs of $F$ then $\Mm$ can simulate both transitions simultaneously. This means that~$\Mm$ can simulate in parallel all possible runs of $\Aa$ over the given input word $\alpha$. The only situation when that fails is a~\emph{collision}, i.e.~the situation when two\footnote{We do not assume that the configurations $(q',\tau')$ and $(q',\tau'')$ are distinct, it might be the case that $\tau'=\tau''$ but there are two distinct runs reaching the configuration $(q',\tau')$.} configurations $(q',\tau')$ and $(q',\tau'')$ can be reached via distinct finite runs over a joint prefix of $\alpha$, with both $(q',\tau')$ and $(q',\tau'')$ belonging to the same club $[q',\gamma']\in F$.

The crucial observation that makes the construction of $\Mm$ possible is Lemma~\ref{lem:unique} --- whenever a collision occurs, it means that either $(q',\tau')\neq (q',\tau'')$ and the club $[q',\gamma']$ is trivial (i.e. $\lang\big(\Aa,[q',\gamma']\big)=\emptyset$); or $(q',\tau')= (q',\tau'')$ and $\lang\big(\Aa,(q',\tau')\big)=\emptyset$. Thus, the following remark holds.

\begin{remark}
\label{rem:collision-empty}
If a configuration $(q',\tau')$ is a part of a collision then $\lang\big(\Aa,(q',\tau')\big)=\emptyset$.
\end{remark}

Let us explain the construction more formally. Assume that the control state of $\Mm$ remembers for which clubs $[q,\gamma]\in F$ the counters $\Cc_1^{[q,\gamma]},\ldots, \Cc_k^{[q,\gamma]}$ actually represent a configuration of $\Aa$ (otherwise their values are irrelevant). Thus, the set of states of $\Mm$ is the set of bitmaps $2^F$ that mark some clubs $[q,\gamma]\in F$ as \emph{inhabited} and the other as \emph{not inhabited}. The initial configuration of $\Mm$ stores $0$ in all the counters and the bitmap marks only one club as inhabited --- the one containing $(q_0,0,\ldots,0)$, i.e.~the initial configuration of~$\Aa$.

Reading a letter $a\in\Sigma$, the machine $\Mm$ simulates all the possible transitions $a: (q,\tau) \mapsto_{\Aa}  (q',\tau')$ for all the configurations $(q,\tau)$ represented in the counters $\Cc_1^{[q,\gamma]},\ldots, \Cc_k^{[q,\gamma]}$ for each inhabited club $[q,\gamma]\in F$. Consider the case when a collision occurs, and two configurations $(q',\tau'), (q',\tau'')\in [q',\gamma']$ belonging to the same club $[q',\gamma']\in F$ need to be stored in the counters $\Cc_1^{[q',\gamma']},\ldots, \Cc_k^{[q',\gamma']}$. In that case $\Mm$ \emph{discards} both configurations $(q',\tau')$ and $(q',\tau'')$.

The control state of $\Mm$ is updated accordingly to know which clubs are inhabited: a club $[q',\gamma']\in F$ is inhabited after such a transition iff at least one of the simulated transitions led to a configuration $(q',\tau')\in [q',\gamma']$ that was not discarded. This concludes the definition of the transition function of $\Mm$.

\subsection{Acceptance condition and equivalence}

We now need to define the acceptance condition of $\Mm$. Assume that an $\omega$\=/word $\sigma=a_1a_2\ldots$ has been read. Consider a run $\runA$ of $\Aa$ over $\sigma$. The consecutive transitions used in $\runA$ are simulated by $\Mm$ when reading $\sigma$. There are two possibilities: either one of the configurations in $\runA$ is discarded in $\Mm$ because of a collision; or all the configurations in $\runA$ are simulated by $\Mm$ (i.e. none of them is discarded). In the latter case we say that $\runA$ is \emph{simulated}.

Let the acceptance condition of $\Mm$ be chosen in such a way that a run of $\Mm$ is accepting if and only if it there exists at least one accepting run $\runA$ of $\Aa$ that is simulated. Notice that Remark~\ref{rem:collision-empty} together with the policy of discarding configurations upon collisions imply the following fact.

\begin{fact}
\label{ft:no-discard}
If $\runA$ is an accepting run of $\Aa$ over an $\omega$\=/word $\sigma$ then it is simulated, i.e.~none of its configurations is discarded by $\Mm$.
\end{fact}

\begin{lemma}
\label{lem:is-regular}
One can encode the above acceptance condition of $\Mm$ as a Muller condition, at the cost of extending the state space of $\Mm$.
\end{lemma}

\begin{proof}
Consider a transition taken by $\Mm$ from one of its states upon reading a letter $a$. Our goal is to encode that transition as an element $\theta$ of $2^{F\times F}$. We will call $\theta$ the \emph{graph} of that transition of $\Mm$. Let the graph $\theta$ contain a pair $\langle[q,\gamma],[q',\gamma']\rangle\in F\times F$ if and only if the club $[q,\gamma]$ was inhabited at the beginning of the transition by a configuration $(q,\tau)\in [q,\gamma]$, $\Mm$ simulated a transition $a : (q,\tau)\mapsto_{\Aa}  (q',\tau')$, the configuration $(q',\tau')$ was not discarded, and $(q',\tau')\in [q',\gamma']$.

Each infinite execution of $\Mm$ while reading an $\w$\=/word $\sigma$ defines a sequence of graphs $\theta_1,\theta_2,\ldots$ encoding the successive transitions of $\Mm$. The $\w$\=/word $\theta_1\theta_2\ldots\in \big(2^{F\times F}\big)^\omega$ is called the \emph{graph} of $\sigma$. A \emph{path} in such a graph is a sequence of clubs $\big([q_i,\gamma_i]\big)_{i=1,\ldots}$ such that for each $i\geq 1$ we have $\langle[q_i,\gamma_i],[q_{i+1},\gamma_{i+1}]\rangle\in \theta_i$. Such a path is \emph{accepting} if infinitely many of the states $q_i$ are accepting.

Notice that (for a fixed $\w$\=/word $\sigma$) there is a bijection between runs $\runA$ of $\Aa$ over $\sigma$ that are simulated and paths in the graph of $\sigma$. Moreover, accepting runs correspond to accepting paths. Therefore, to check if there is an accepting simulated run it is enough to check if the graph $\theta_1\theta_2\ldots\in \big(2^{F\times F}\big)^\omega$ contains an accepting path.

Treating $2^{F\times F}$ as a finite alphabet, the $\w$\=/language of $\w$\=/words $\theta_1\theta_2\ldots\in \big(2^{F\times F}\big)^\omega$ that contain an accepting path is regular and therefore it can be recognised by a deterministic Muller automaton $\mathcal{D}$. By extending each state of $\Mm$ by an additional state of $\mathcal{D}$, one can encode the acceptance condition of $\mathcal{D}$ directly on $\Mm$, turning it into a Muller machine.
\end{proof}

Notice that the above lemma relies heavily on the fact that the acceptance condition of $\Aa$ is encoded only on the set of states $K$ of $\Aa$, without considering the counter values. This allows us to detect accepting runs of $\Aa$ by looking only on the sequence of clubs that were visited.

This concludes the construction of the machine $\Mm$. The following corollary of Fact~\ref{ft:no-discard} implies that $\Mm$ satisfies the requirements from Theorem~\ref{thm:main-unamb}.

\begin{corollary}
\label{cor:correct}
The machine $\Mm$ accepts an $\omega$\=/word $\sigma\in\Sigma^\omega$ if and only if $\sigma\in\lang(\Aa)$.
\end{corollary}

\subsection{Discussion}

The above construction is quite flexible both in terms of the input and output models (i.e.~the exact abilities of the automata $\Aa$ and $\Mm$ respectively). For instance, if one doesn't care about the effectiveness of the determinisation construction, one can extend it to allow the input automaton to \emph{reset} the value of some of its counters to~$0$ (this can be simulated by $\Mm$ using copying). Clearly Remark~\ref{rem:residual-order} is still valid for these machines, therefore Remark~\ref{rem:subclub-noneff} applies one can repeat the rest of the argument proving the existence of the machine $\Mm$. Notice that one cannot hope for an effective variant of Remark~\ref{rem:subclub-noneff} (i.e.~Lemma~\ref{lem:club-non-empty}) because the non-emptiness problem for B\"uchi blind counter automata with resets is undecidable (see Theorem~10 and Lemma~17 in~\cite{mayr_undacidable_lossy}).

On the other hand, it seems that the ability of the output machine $\Mm$ to perform zero tests is inherent to that construction. At the same time, $\Aa$ needs to be blind because otherwise it would violate Remark~\ref{rem:residual-order}. This implies that there might be no single type of counter machines so that both $\Aa$ and $\Mm$ could be assumed to be of that type (which would provide a determinisation construction within that type of automata). Also, there seems to be no reasonable way to avoid the need of copying the counters in~$\Mm$, as the graph of the possible traces in $\Mm$ can be complex and branching.

The above construction is arranged in such a way to ensure that $\Mm$ is a real time multicounter automaton with zero tests and counter copying. This model of machines is much more concrete than general Turing machines, but in the end its expressive power is essentially the same --- already two\=/counter Minsky machines with zero tests are able to simulate arbitrary Turing machines~\cite{minsky_two_counter}.
\end{document}